\pdfoutput=1
\documentclass[preprint,10pt]{sigplanconf} 
\usepackage[english]{babel}
\usepackage[T1]{fontenc}
\usepackage{amsbsy,amscd,amsfonts,amssymb,amstext,amsmath,amsopn,latexsym,amsthm} 
\usepackage[utf8]{inputenc}
\usepackage{graphicx}
\usepackage{url}
\usepackage{semantic} 
\usepackage{xifthen} 
\usepackage{color}
\usepackage{scalerel}
\usepackage{etoolbox}
\usepackage{microtype} 
\usepackage{flushend}

\usepackage{listings}

\lstdefinelanguage{json}
{
  basicstyle=\scriptsize,  	
  morestring=[b]",
	morestring=[b]',
	morestring=[b]"""    
commentstyle=\color{grey},
    showstringspaces=false,
}

\usepackage[disable]{todonotes}
\usepackage{cleveref}
\usepackage{bcprules}
\usepackage{xspace}
\usepackage[hidelinks]{hyperref}
\usepackage{enumitem}
\usepackage{multicol}

\theoremstyle{plain}
\newtheorem{theorem}{Theorem}
\newtheorem{proposition}{Proposition}
\newtheorem{lemma}{Lemma}
\newtheorem{definition}{Definition}

\crefname{subsection}{section}{sections}
\crefname{subsubsection}{section}{sections}

\newcommand{\lang}{{\textsf{\textmd{CPL}}}\xspace}

\renewcommand{\floatpagefraction}{.8}%

\lstdefinelanguage{Scala}{
	morekeywords={},
	morekeywords=[2]{
		public, private, protected,
		abstract,case,catch,class,def,
		do,else,extends,false,final,finally,
		for,if,implicit,import,match,mixin,
		new,null,object,override,package,
		private,protected,requires,return,sealed,
		super,this,throw,trait,true,try,
		type,val,var,while,with,yield,
		lazy,evt,observable,imperative,
                after,before},
	otherkeywords={=>,<-,<\%,<:,>:,\#,@},
	sensitive=true,
	morecomment=[l]{//},
	morecomment=[n]{/*}{*/},
	morestring=[b]",
	morestring=[b]',
	morestring=[b]"""
}
\lstdefinelanguage{cpl}{
          morekeywords={
                  this,snap,repl,let,letk,case,
                  matchk,in,as,thunk,spwn,srv,inst,
                  img,par,if,then,else},
          morekeywords=[2]{Top,Unit,Int,Dbl,String,Pair,Map,List},
          sensitive=true,
          columns=fullflexible,
          keepspaces=true,
          morecomment=[l]{//},
          morecomment=[n]{/*}{*/},
          moredelim=[is][\sffamily]{!}{!},
          moredelim=[is][\fontfamily{cmr}\itshape]{'}{'},
          basicstyle=\fontfamily{cmtt}\small,
          keywordstyle=\fontfamily{cmr}\bfseries,
          keywordstyle=[2]\sffamily,
          escapeinside={(*@}{@*)},
          literate={\%bot}{{$\bot$}}1 {[}{{$[$}}1 {]}{{$]$}}1 {\%w}{{$\omega$}}1 {\%a}{{$\alpha$}}1 {\%All}{{$\forall$}}1 {\%b}{{$\beta$}}1 {\%l}{{$\lambda$}}1 {\%L}{{$\Lambda$}}1 {\%y}{{$\gamma$}}1 {\%k}{{$\kappa$}}1 {\%v}{{$\nu$}}1 {\%e}{{$\varepsilon$}}1 {||}{{$\parallel$}}1 {!=}{{$\neq$}}1 {\#}{{$\sharp$}}1 {,}{{$,$}}1 {.}{{$.$}}1  {(}{{$($}}1 {)}{{$)$}}1 {\{}{{$\{$}}1 {\}}{{$\}$}}1  {~>}{{$\rightsquigarrow$}}1  {<}{{$\langle$}}1 {<:}{{$\subtpe$}}1 {>}{{$\rangle$}}1 {<=}{{$\leq$}}1 {>=}{{$\geq$}}1 {:}{{$:$}}1 {:>}{{$\triangleright$}}1  {->}{{$\to$}}1
}
\let\origthelstnumber\thelstnumber
\makeatletter
\newcommand*\Suppressnumber{%
  \lst@AddToHook{OnNewLine}{%
    \let\thelstnumber\relax%
     \advance\c@lstnumber-\@ne\relax%
    }%
}

\newcommand*\Reactivatenumber{%
  \lst@AddToHook{OnNewLine}{%
   \let\thelstnumber\origthelstnumber%
   \advance\c@lstnumber\@ne\relax}%
}

\lstnewenvironment{codenv}{\lstset{language=Scala,numbers=left,xleftmargin=1.5em}}{}
\lstnewenvironment{cplcode}{\lstset{language=cpl}}{}
\newcommand{\cplinline}[1]{\lstinline[language=cpl,mathescape=true,basicstyle=\fontfamily{cmtt}\normalsize]|#1|}

\def\tecreport{}

\newcommand{\figfontsize}{\footnotesize}

\mathlig{<}{\langle}
\mathlig{>}{\rangle}
\newcommand{\gt}{\mathligsoff > \mathligson}
\newcommand{\lt}{\mathligsoff < \mathligson}

\newcommand{\Set}[1]{\ensuremath{\left\{#1\right\}}} 
\newcommand{\dom}[1]{\ensuremath{\operatorname{dom}(#1)}} 
\newcommand{\mapadd}[2]{\ensuremath{#1[#2]}} 

\newcommand{\Nat}{\ensuremath{\mathbb{N}}}
\newcommand{\Names}{\ensuremath{\mathcal{N}}}

\newcommand{\Subst}[1]{\ensuremath{\Set{#1}}}

\newcommand{\tuple}[1]{\ensuremath{\overline{#1}}}

\newcommand{\repeatJdgS}[2]{\forall#1.\ \ #2}
\newcommand{\existsJdg}[2]{\exists#1.\ \ #2}

\reservestyle[\mathinner]{\jckeyword}{\mathbf}
\jckeyword{let[let\;],letk[letk\;],in[\;in\;],this,as[\;as\;],spwn[spwn\;],inst[inst\;],par[par\;],srv[srv\;],cont,thunk[thunk\;],snap[snap\;],repl[repl\;],img[img\;],matchk[matchk\;],case[case\;]}
\newcommand{\repl}[2]{\ensuremath{\<repl>#1\;#2}}
\newcommand{\dead}{\ensuremath{\mathbf{0}}}

\newcommand{\srvt}[1]{\ensuremath{\<srv> #1}}

\newcommand{\svc}[2]{\ensuremath{#1\sharp#2}}
\newcommand{\joined}[0]{\ensuremath{\mathop{\&}}}
\newcommand{\prlll}[1]{\ensuremath{\<par> #1}}

\newcommand{\trr}{\ensuremath{\triangleright}}

\newcommand{\CPSname}{\ensuremath{\mathsf{T}}\xspace}
\newcommand{\CPS}[2]{\ensuremath{\CPSname(#1, #2)}}

\newcommand{\Tuni}[2]{\ensuremath{\forall #1.\; #2}}
\newcommand{\Tsvc}[1]{\ensuremath{<\,#1\,>}}
\newcommand{\Tsrv}[1]{\ensuremath{\<srv> #1 }}
\newcommand{\Tsrvinst}[1]{\ensuremath{\<inst> #1}}
\newcommand{\Type}[1]{\ensuremath{\mathsf{#1}}}
\newcommand{\Tapp}[2]{\ensuremath{#1\;[#2]}}
\newcommand{\Tabs}[2]{\ensuremath{\Lambda #1.\, #2}}
\newcommand{\Top}{\Type{Top}}
\newcommand{\Unit}{\Type{Unit}}
\newcommand{\Tfun}[2]{\ensuremath{#1\to #2}}
\newcommand{\tpe}{\colon}
\newcommand{\Timg}[1]{\ensuremath{\<img> #1}}
\newcommand{\Tjudge}[3]{\ensuremath{#1 \vdash\nobreak #2 : \nobreak #3}}

\newcommand{\ftv}[1]{\ensuremath{\functionsymbol{ftv}{#1}}}

\newcommand{\functionsymbol}[2]{\ensuremath{\operatorname{#1}\ifthenelse{\isempty{#2}}{}{(#2)}}}

\newcommand{\context}[2]{\ensuremath{#1\ifthenelse{\isempty{#2}}{}{[#2]}}}
\newcommand{\hole}{\ensuremath{[\cdot]}}

\newcommand{\matchp}[2]{\ensuremath{\functionsymbol{match}{#1}\Downarrow(#2)}}

\newcommand{\oto}[1]{\ensuremath{1 \ldots #1}}

\newcommand{\subdereq}{\ensuremath{\mathbin{\preceq}}}

\newcommand{\subtpe}{\ensuremath{\mathbin{\!\scalerel{{\lt}\!}{\colon}\!\!}}}
\newcommand{\Subjudge}[3]{\ensuremath{#1\vdash#2 \subtpe #3}}

\newcommand{\id}[1]{\ensuremath{\mathit{#1}}}

\DeclareMathSymbol{\mlq}{\mathord}{operators}{``}
\DeclareMathSymbol{\mrq}{\mathord}{operators}{`'}

\makeatletter 
\newcommand{\negphantom}{\v@true\h@true\negph@nt} 
\newcommand{\neghphantom}{\v@false\h@true\negph@nt} 
\newcommand{\negph@nt}{\ifmmode\expandafter\mathpalette 
  \expandafter\mathnegph@nt\else\expandafter\makenegph@nt\fi} 
\newcommand{\makenegph@nt}[1]{%
  \setbox\z@\hbox{\color@begingroup#1\color@endgroup}\finnegph@nt} 
\newcommand{\finnegph@nt}{%
  \setbox\tw@\null 
  \ifv@ \ht\tw@\ht\z@\dp\tw@\dp\z@\fi \ifh@\wd\tw@-\wd\z@\fi\box\tw@} 
\newcommand{\mathnegph@nt}[2]{%
  \setbox\z@\hbox{$\m@th #1{#2}$}\finnegph@nt} 
\makeatother

\newcommand{\implementation}{\url{https://github.com/seba--/djc-lang}}

\begin{document}
\setlength{\abovedisplayskip}{0pt}
\setlength{\belowdisplayskip}{0pt}
\setlength{\abovedisplayshortskip}{0pt}
\setlength{\belowdisplayshortskip}{0pt}

\title{\lang: A Core Language for Cloud Computing}
\ifdef{\tecreport}{\subtitle{Technical Report for the Conference Publication~\cite{BracevacConf}}}{}
\authorinfo
{
Oliver Bra\v{c}evac$^{1}$ \and
Sebastian Erdweg$^{1}$\and
Guido Salvaneschi$^{1}$\and
Mira Mezini$^{1,2}$}
{$^1$TU Darmstadt, Germany\quad $^2$Lancaster University, UK}

\maketitle

\begin{abstract}
Running distributed applications in the cloud involves deployment.
That is, distribution and configuration of application services and middleware infrastructure. 
The considerable complexity of these tasks resulted in 
the emergence of declarative JSON-based domain-specific \emph{deployment languages}
to develop \emph{deployment programs}.
However, existing deployment programs unsafely compose artifacts written in different languages, 
leading to bugs that are hard to detect before run time. Furthermore,
deployment languages do not provide extension points for custom implementations of existing cloud services such as application-specific load balancing policies.

To address these shortcomings, we propose 
\lang{} (\emph{Cloud Platform Language}), a statically-typed core 
language for programming both distributed applications as well as 
their deployment on a cloud platform. In \lang{}, application services and
deployment programs interact through statically typed, extensible interfaces,
and an application can trigger further deployment at run time.
We provide a formal semantics of \lang{}
and demonstrate that it enables type-safe, composable and extensible libraries
of \emph{service combinators}, such as load balancing and fault tolerance.


\end{abstract} 

\ifdef{\tecreport}{}{


}

\section{Introduction}
\label{sec:introduction-1}

Cloud computing~\cite{Vaquero:2008:BCT:1496091.1496100} has emerged as the reference
infrastructure for concurrent distributed services with high availability, resilience
and quick response times, providing access
to on-demand and location-transparent computing resources. Companies
develop and run distributed applications on specific cloud platforms, e.g., 
Amazon AWS\footnote{\url{https://aws.amazon.com}} or
Google Cloud Platform.\footnote{\url{https://cloud.google.com}}
Services are bought as needed from the cloud provider in order to
adapt to customer demand,

An important and challenging task in the development process of cloud applications is \emph{deployment}.
Especially, deployment involves the distribution, configuration and composition of
(1)  virtual machines that implement the application and its services,
and of
(2) virtual machines that provide middleware infrastructure such as load balancing, key-value stores, and MapReduce.
Deploying a cloud application can go wrong and cause the application to malfunction. 
Possible causes are software bugs in the  application itself, but also
wrong configurations, such as missing library dependencies or inappropriate 
permissions for a shell script.
Fixing mistakes \emph{after} deployment causes high costs and loss of reputation.
For example, in 2012, Knight Capital lost over \$440 Million over the course of 30 minutes due to 
a bug in its deployed trading software,\footnote{\url{http://www.bloomberg.com/bw/articles/2012-08-02/knight-shows-how-to-lose-440-million-in-30-minutes}.} causing the disappearance of the company from the market.

Considering that cloud applications can have deployment sizes in the hundreds or thousands of 
virtual machines, manual configuration is error-prone and does not scale. Cloud platforms address this issue with
domain-specific languages (DSLs) such as Amazon CloudFormation or Google Cloud Deployment Manager. 
The purpose of these DSLs is to write reusable \emph{deployment programs}, which instruct the cloud platform 
to perform deployment steps automatically. A typical deployment program
specifies the required virtual machines for the application infrastructure, how these virtual
machines are connected with each other, and how the application infrastructure connects to the
pre-existing or newly created middleware infrastructure of the cloud platform.

However, the modularity of current cloud deployment DSLs is insufficient (detailed discussion in Section~\ref{sec:motivation}):

{\bf Unsafe Composition:} Application services and deployment programs are written
  in different languages. Deployment
  DSLs configure application services by lexically expanding configuration
  parameters into application source code before its execution.
  This approach is similar to a lexical macro system and makes deployment
  programs unsafe because of unintentional code injection and lexical incompatibilities.

{\bf No Extensibility:} 
Middleware cloud services (e.g., elastic load balancing, which may
dynamically allocate new virtual machines) are pre-defined in the cloud platform
and only referenced by the deployment program through external interfaces. 
As such, there is no way to customize those services during deployment or extend them with additional features.

{\bf Stage Separation:} Current deployment DSLs finish their execution before
the application services are started. Therefore, it is impossible to change the
deployment once the application stage is active.  Thus,
applications cannot self-adjust their own deployment, e.g., to react to time-varying
customer demand.

We propose \lang{} (Cloud Platform Language), a statically-typed core language 
for programming cloud applications and deployments. 
\lang{} employs 
techniques from programming language design and type systems to
overcome the issues outlined above.
Most importantly, \lang{} unifies the programming of deployments and applications into a single language.
This avoids unsafe composition because deployments and applications can exchange values directly via statically typed interfaces.
For extensibility, \lang{} supports higher-order service combinators with statically typed interfaces using bounded polymorphism.
Finally, \lang{} programs run at a single stage where an application service can trigger further deployment.

To demonstrate how \lang{} solves the problems of deployment languages,
we implemented a number of case studies. First, we demonstrate type-safe composition 
through generically typed worker and thunk abstractions.
Second, on top of the worker abstraction, we define composable and reusable \emph{service combinators} in \lang{},
which add new features, such as elastic load balancing and fault tolerance.
Finally, we demonstrate how to model MapReduce as a deployment program in \lang{} and apply our
combinators, obtaining different MapReduce variants, which safely deploy at run time.

In summary, we make the following contributions:
\begin{itemize}[topsep=.1ex]
\item We analyze the problems with current cloud deployment DSLs.
\item We define the formal syntax and semantics of \lang{} to model
  cloud platforms as distributed, asynchronous message passing
  systems.  Our design is inspired by the Join
  Calculus~\cite{Fournet-Gonthier:popl96}.
\item We define the type system of \lang{} as a variant of System F with bounded quantification~\cite{pierce2002types}.
\item We formalize \lang{} in PLT Redex~\cite{FelleisenFF09Redex} and we
  provide a concurrent implementation in Scala.
\item We evaluated \lang{} with case studies, including a library of
  typed service combinators that model elastic load balancing and
  fault tolerance mechanisms. Also, we apply the combinators to a MapReduce
  deployment specification.
\end{itemize}

\noindent
The source code of the PLT Redex and Scala implementations and of all
case studies is available online:\linebreak \implementation.


\section{Motivation}
\label{sec:motivation}

In this section, we analyze the issues that programmers encounter with current
configuration and deployment languages on cloud platforms by a concrete example.

\begin{figure}
  \centering
\begin{codenv}
{ //... 
 "Parameters": { (*@\label{param0}@*)
  "InstanceType": { (*@\label{param}@*)
   "Description": "WebServer EC2 instance type", 
   "Type": "String",
   "Default": "t2.small",
   "AllowedValues": [ "t2.micro", "t2.small ]",
   "ConstraintDescription": "a valid EC2 instance type."
  }  //...
 },   (*@\label{paramn}@*)
 "Resources": {  (*@\label{resources0}@*)
   "WebServer": {  (*@\label{ws}@*)
     "Type": "AWS::EC2::Instance",  (*@\label{badtype}@*)
     "Properties": {
       "InstanceType": { "Ref" : "InstanceType" }  (*@\label{instype}@*),
       "UserData": { "Fn::Base64" : { "Fn::Join" : ["", [
             "#!/bin/bash -xe\n",  (*@\label{script0}@*)
             "yum update -y aws-cfn-bootstrap\n",

             "/opt/aws/bin/cfn-init -v ",
             "   --stack ", { "Ref" : "AWS::StackName" },
             "   --resource WebServer ",
             "   --configsets wordpress_install ",
             "   --region ", { "Ref" : "AWS::Region" }, "\n" (*@\label{scriptn}@*)
       ]]}}, //...
     },
   }
 },     (*@\label{resourcesn}@*)
 "Outputs": {   (*@\label{out0}@*)
   "WebsiteURL": {
     "Value": 
       { "Fn::Join" : 
        ["", ["http://", { "Fn::GetAtt" : 
         [ "WebServer", "PublicDnsName" ]}, "/wordpress" ]]},
     "Description" : "WordPress Website"
   }
 }   (*@\label{outn}@*)
} (*@\label{end}@*)
\end{codenv}
\caption{A deployment program in CloudFormation (details omitted, full version:
\scriptsize\url{https://s3.eu-central-1.amazonaws.com/cloudformation-templates-eu-central-1/WordPress_Multi_AZ.template}.}
\vspace{-2mm}
\label{fig:cloudformationexample}
\end{figure}

\subsection{State of the Art}

Figure~\ref{fig:cloudformationexample} shows an excerpt of a
deployment program in CloudFormation, a JSON-based DSL for Amazon AWS. 
The example is from the CloudFormation documentation.
We summarize the main characteristics of the deployment language
below.

\begin{itemize}

\item Input parameters capture varying details of a configuration
  (Lines~\ref{param0}-\ref{paramn}). For example, the program
  receives the virtual machine instance type that should host the web server for a user blog (Line~\ref{param}). 
  This enables reuse of the program with different parameters.

\item CloudFormation programs specify {\it named resources} to be created in the
  deployment (Lines~\ref{resources0}-\ref{resourcesn}), e.g., deployed
  virtual machines, database instances, load balancers and even other
  programs as modules. The program in Figure~\ref{fig:cloudformationexample}
  allocates a "WebServer" resource (Line~\ref{ws}), which is a virtual machine instance.
  The type of the virtual machine references a parameter (Line~\ref{instype}), that
the program declared earlier on (Line~\ref{param}). Resources can
  refer to each other, for example, configuration parameters of a web
  server may refer to tables in a database resource.

\item Certain configuration phases require to execute application code inside
  virtual machine instances after the deployment stage.
  Application code is often directly specified 
  in resource bodies (Lines~\ref{script0}-\ref{scriptn}). In the example, a
  bash script defines the list of software packages to install on the
  new machine instance (in our case a WordPress\footnote{\url{http://wordpress.org}} blog). 
  In principle, arbitrary programs in any language can be specified.

\item Deployment programs specify output parameters
  (Lines \ref{out0}-\ref{outn}), which may depend on input parameters
  and resources. Output parameters are returned to the caller after executing the
  deployment program. In this example, it is a URL pointing to the new 
  WordPress blog.

\item The deployment program is interpreted at run time by the cloud platform
  which performs the deployment steps according to the specification.

\end{itemize}

\subsection{Problems with Deployment Programs}

In the following we discuss the issues with the CloudFormation example described
above. 

\begin{description}

\item[Internal Safety] Type safety for deployment programs is
  limited. Developers define ``types'' for resources of the cloud
  platform, e.g., {\tt AWS::EC2::Instance} (Line~\ref{badtype})
  represents an Amazon EC2 instance. However, the typing system of
  current cloud deployment languages is  primitive and only
  relies on the JSON types. 

\item[Cross-language Safety] Even more problematic are issues caused
  by cross-language interaction between the deployment language and
  the language(s) of the deployed application services. For example, the {\tt AWS::Region}
  variable is passed from the JSON specification to
  the bash script (Line~\ref{scriptn}). However, the sharing mechanism
  is just syntactic replacement of the current value of {\tt AWS::Region}
  inside the script. Neither are there type-safety checks nor syntactic checks
  before the script is executed. More generally,
  there is no guarantee that the data types of the deployment language
  are compatible with the types of the application language nor that
  the resulting script is syntactically correct. This problem makes
  cloud applications susceptible to hygiene-related bugs and injection attacks~\cite{Bravenboer:2007aa}.

\item[Low Abstraction Level] Deployment languages
typically are Turing-complete but the abstractions are low-level and not deployment-specific.
For example, (1) deployment programs receive parameters and 
return values similar to procedures and (2) deployment programs can be
instantiated from inside other deployment programs, which resembles modules.
Since deployment is a complex engineering task, advanced language features
are desirable to facilitate programming in the large, e.g., 
higher-order combinators, rich data types and strong interfaces.

\item[Two-phase Staging]  Deployment programs in current DSLs execute before
the actual application services, that is, information flows from the deployment language
 to the deployed application services, but not the other way around. 
 As a result, an application service cannot reconfigure a deployment based on the
 run time state. Recent scenarios in reactive and big data computations
 demonstrate that this is a desirable feature~\cite{Fernandez:2013fb}.

\item[Lack of Extensibility] Resources and service references in
deployment programs refer to pre-defined abstractions of the cloud platform,
which have rigid interfaces. Cloud platforms determine the semantics
of the services. Programmers cannot implement their own variants of services that
plug into the deployment language with the same interfaces as the native services.

\item[Informal Specification] The behavior of JSON deployment scripts is
  only informally defined. The issue is exacerbated by the
  mix of different languages. As a result, it is hard
  to reason about properties of systems implemented using deployment
  programs.

\end{description}

The issues above demand for a radical
change in the way programmers deploy cloud applications and in the way
application and deployment configuration code relate to each other.


\section{The Cloud Platform Language}
\label{sec:cloud-calculus}

A solution to the problems identified in the previous section requires an holistic approach
where cloud abstractions are explicitly represented in the language.
Programmers should be able to modularly specify application behavior as well as
reconfiguration procedures. Run time failures should be prevented at
compilation time through type checking.

These requirements motivated the design of \lang{}. 
In this section, we present its syntax and the operational semantics.
\subsection{Language Features in a Nutshell}
\label{sec:design-principles}

\begin{description}
\item[\textbf{Simple Meta-Theory}:]  \lang{} should serve as the basis for
  investigating high-level language features and type systems designed for cloud
  computations. To this end, it is designed as a core language with a simple meta-theory.
  Established language features and modeling techniques, such as lexical scoping and a
  small-step operational semantics, form the basis of \lang{}.

\item[\textbf{Concurrency}:]  \lang{} targets distributed
   concurrent computations.
    To this end, 
    it allows the definition of independent computation units, which we call
  \emph{servers}. 

\item[\textbf{Asynchronous Communication}:] Servers can receive parameterized
  \emph{service requests} from other servers. To realistically model low-level
  communication within a cloud, the language only provides asynchronous
  end-to-end communication, where the success of a service request is not
  guaranteed. Other forms of communication, such as synchronous, multicast, or
  error-checking communication, can be defined on top of the
  asynchronous communication.  

\item[\textbf{Local Synchronization}:] Many useful concurrent and asynchronous applications
  require synchronization. We adopt \emph{join patterns} 
  from the Join Calculus~\cite{Fournet-Gonthier:popl96}. Join patterns 
  are \emph{declarative} synchronization primitives for machine-local synchronization.          
    
\item[\textbf{First-Class Server Images}:]
    Cloud platforms employ virtualization to spawn, suspend and duplicate 
  virtual machines. That is, virtual machines are data that can be stored and send as payload in messages. This is the core idea behind \lang{}'s design and enables
  programs to change their deployment at run time.
  Thus, \lang{} features servers as values, called \emph{first-class servers}.
  Active \emph{server instances} consist of an address, which points
  to a \emph{server image} (or \emph{snapshot}). The server image embodies
  the current run time state of a server and a description of the server's functionality, which we call \emph{server template}. At run
  time, a server instance may be overwritten by a new server image, thus changing the behavior for subsequent service requests to that instance.

\item[\textbf{Transparent Placement}:] Cloud platforms can reify new machines
  physically (on a new network node) or virtually (on an existing network
  node). Since this difference does not influence the semantics of a program but
  only its non-functional properties (such as performance), our semantics is
  transparent with respect to placement of servers. Thus, actual languages based
  on our core language can freely employ user-defined placement definitions and automatic
  placement strategies. Also, we do not require that \lang{}-based languages map servers to virtual machines, which may be inefficient for 
  short-lived servers. 
  Servers may as well represent local computations executing on a virtual machine.
\end{description}

\subsection{Core Syntax}
\label{sec:syntax}

Figure~\ref{fig:djcsyntax} displays the core syntax of \lang{}. An
expression $e$ is either a value or one of the following syntactic
forms:\footnote{We write $\tuple{a}$ to denote the finite sequence
  $a_1\!\ldots{}a_n$ and we write $\varepsilon$ to denote the empty sequence.}
\begin{itemize}
\item A \emph{variable} $x$ is from the countable set $\Names$. Variables
  identify services and  parameters of their
   requests.

\item A \emph{server template} $(\srvt{\tuple{r}})$ is a first-class value that describes the behavior of a
  server as a sequence of reaction rules $\tuple{r}$. A reaction rule takes
  the form $\tuple{p} \triangleright e$, where $\tuple{p}$ is a sequence of
  joined service patterns and $e$ is the body. A \emph{service pattern}
  $x_0<\tuple{x}>$ in $\tuple{p}$ declares a service named $x_0$ with parameters
  $\tuple{x}$ and a rule can only fire if all service patterns are matched
  simultaneously. The same service pattern can occur in multiple rules of a
  server.

\item A \emph{server spawn} $(\<spwn> e)$ creates a new running \emph{server
    instance} at a freshly allocated \emph{server address} $i$ from a 
  given \emph{server image} $(\srvt{\tuple{r}}, \tuple{m})$
  represented by $e$. A \emph{server image} is a description of a server behavior
  plus a server state -- a \emph{buffer} of pending messages.
  A real-world equivalent of server images are e.g., virtual machine snapshots.
    A special case of a server image is the value $\dead$,
  which describes an inactive or shut down server.

\item A fully qualified \emph{service reference} $\svc{e}{x}$, where $e$ denotes
  a server address and $x$ is the name of a service provided by the server instance at $e$. Service
  references to server instances are themselves values.

\item A \emph{self-reference} $\<this>$ refers to the address of the lexically
  enclosing server template, which, e.g., allows one service to call upon other
  services of the same server instance.

\item An asynchronous \emph{service request} $e_0<\tuple{e}>$, where  $e_0$
  represents a service reference and $\tuple{e}$ the arguments of the
  requested service.

\item A \emph{parallel expression} $(\prlll{\tuple{e}})$ of service requests
  $\tuple{e}$ to be executed independently. The empty parallel expression
  $(\prlll{\varepsilon})$ acts as a noop expression, unit value, or null process
  and is a value.

\item A \emph{snapshot} $\<snap> e$ yields an image of the server instance which resides at the address denoted by $e$.

\item A \emph{replacement} $\repl{e_{1}}{e_{2}}$ of the server instance at address
$e_{1}$ with the server image $e_{2}$.
\end{itemize}

\begin{figure*}[t]\centering
$\figfontsize\begin{array}{p{.5\textwidth}@{\hspace{0.8cm}}p{.5\textwidth}}
\begin{array}[t]{l@{\hspace{1ex}}l@{\hspace{1ex}}l@{\hspace{2em}}l}
e             & ::= &  v \mid  x \mid  \<this> \mid  \srvt{\tuple{r}} \mid  \<spwn> e \mid \svc{e}{x} \mid  e<\tuple{e}> \mid\prlll{\tuple{e}} & \text{(Expressions)} \\
              &     &  \phantom{v} \mid \<snap> e \mid \repl{e}{e} &  \\[2ex]
v             & ::= &  \srvt{\tuple{r}} \mid i \mid \svc{i}{x} \mid \prlll{\varepsilon} \mid (\srvt{\tuple{r}}, \tuple{m}) \mid \dead &  \text{(Values)} \\[2ex]
\context{E}{} & ::= &  \hole \mid  \<spwn> \context{E}{} \mid \svc{\context{E}{}}{x} \mid \context{E}{}<\tuple{e}> \mid  e<\tuple{e}\;\context{E}{}\;\tuple{e}>  & \text{(Evaluation Contexts)}\\
              &     & \phantom{\hole} \mid  \prlll{\tuple{e}\; \context{E}{}\; \tuple{e}} \mid \<snap> \context{E}{} \mid \repl{\context{E}{}}{e} \mid \repl{e}{\context{E}{}} & \\[2ex]
\multicolumn{3}{l}{x, y, z\in\Names}  & \text{(Variable Names)}
\end{array}
&
\begin{array}[t]{l@{\hspace{1ex}}l@{\hspace{1ex}}l@{\hspace{2em}}l}
\multicolumn{3}{l}{i\in\Nat}          & \text{(Server Addresses)} \\[1ex]
r &::=& \tuple{p} \triangleright e    & \text{(Reaction Rules)} \\[1ex]
p &::=& x<\tuple{x}>                  & \text{(Join Patterns)}  \\[1ex]
m &::=& x<\tuple{v}>                  & \text{(Message Values)} \\[1ex]
\mu &::=& \varnothing \mid \mu ; i\mapsto (\srvt{\tuple{r}}, \tuple{m}) & \text{(Routing Tables)} \\
    &    & \phantom{\varnothing} \mid \mu;i\mapsto \dead &  
\end{array}
\end{array}$
 \caption{Expression Syntax of \lang{}.}
\label{fig:djcsyntax}
\end{figure*}


\begin{figure*}[t]
\figfontsize
\let\oldtextsc\textsc
\let\oldafterrule\afterruleskip
\renewcommand{\rn}[1]{\scriptsize\textsc{#1}}
\renewcommand{\afterruleskip}{2ex}
\smallrulenames
\typicallabel{Match$_1$}
\centering
  \begin{multicols}{2}
  \infrule[Cong]{e\,{\mid}\,\mu \; --> \;
      e'\,{\mid}\,\mu'}{\context{E}{e}\,{\mid}\,\mu \; --> \; \context{E}{e'}\,{\mid}\,\mu'}
\vspace{\afterruleskip} 
    \infrule[Par]{}{\prlll{\tuple{e_1}\; ( \prlll{\tuple{e_2}})\;
        \tuple{e_3}}\,{\mid}\,\mu \; --> \; \prlll{ \tuple{e_1}\; \tuple{e_2}\;
        \tuple{e_3} }\,{\mid}\,\mu} 
\vspace{\afterruleskip} 
    \infrule[Rcv]{\mu(i) = (s, \tuple{m})}{ \svc{i}{x}<\tuple{v}>\,{\mid}\,
      \mu{} \; --> \; \prlll{\varepsilon{}}\,{\mid}\, \mu{}; i \mapsto
        (s, \tuple{m}\cdot x<\tuple{v}>)  }  
\vspace{\afterruleskip} 
    \infrule[React]{ \mu(i) = (s, \tuple{m}) \andalso
      s = \srvt{\tuple{r}_{1}\; (\tuple{p}\triangleright{} e_{b})\; \tuple{r_{2}}} \\
      \matchp{\tuple{p}, \tuple{m}}{\tuple{m}', \sigma} \andalso
      \sigma_b = \sigma \cup \Subst{\<this> := i}} { \prlll{e} \,{\mid}\, \mu \; --> \;
      \prlll{e\ \sigma_b(e_b)} \,{\mid}\, \mu;{i \mapsto (s,
        \tuple{m}') } }
\vspace{\afterruleskip} 
    \infrule[Spwn]{i \notin\dom{\mu{}} \andalso (s = \dead \vee s =
      (\srvt{\tuple{r}}, \tuple{m})) }{\<spwn> s \,{\mid}\,{}\mu{} \; --> \; i
      \,{\mid}\, \mu;{i \mapsto s }}
\vspace{\afterruleskip} 
    \infrule[Snap]{\mu(i) = s \andalso (s = \dead \vee s =
      (\srvt{\tuple{r}}, \tuple{m}))}{\<snap> i \,{\mid}\, \mu \; --> \;
      s\,{\mid}\,\mu{} }
\vspace{\afterruleskip} 
    \infrule[Repl]{i \in{} \dom{\mu{}} \andalso (s = \dead \vee s =
      (\srvt{\tuple{r}}, \tuple{m})) }{ \repl{i}{s} \,{\mid}\,{}\mu{} \; --> \;
      \prlll{\varepsilon}\,{\mid}\,\mu;{i \mapsto s} }
\vspace{\afterruleskip}\vspace{\afterruleskip} 
\vspace{\afterruleskip} 
Matching Rules:
\vspace{\afterruleskip} 
    \infrule[Match$_0$]{
                            }{ \matchp{\varepsilon, \tuple{m}}{\tuple{m}, \emptyset} }
    \vspace{\afterruleskip} 
\infrule[Match$_1$]{
      \tuple{m} = \tuple{m}_{1}\; (x<v_1 \ldots v_k>)\; \tuple{m}_{2} \andalso
      \sigma = \Subst{x_i := v_i \mid 1 \le i \le k}   \\
      \matchp{\tuple{p}, \tuple{m_1}\; \tuple{m_2}}{\tuple{m}_{r}, \sigma_r} \andalso
      \dom{\sigma} \cap \dom{\sigma_r} = \emptyset }{ \matchp{x<x_1
        \ldots x_k >\; \tuple{p}, \tuple{m}}{\tuple{m}_{r}, \sigma
        \cup \sigma_r} }
  \end{multicols}
\vspace{1em}
  \caption{Small-step Operational Semantics of \lang{}.}
\label{fig:djcsemantics}\label{fig:djcmatchingrules}
\renewcommand{\textsc}{\oldtextsc}
\renewcommand{\afterruleskip}{\oldafterrule}
\end{figure*}


\ifdef{\tecreport}{
\begin{figure*}[t]\centering
  $\figfontsize
  \begin{array}{p{.3\textwidth}@{\hspace{0.1cm}}p{.7\textwidth}}
\begin{array}[t]{l@{\hspace{1ex}}l@{\hspace{1ex}}l@{\hspace{2em}}l}
   e  & ::= & \ldots \mid \Tabs{\alpha\subtpe T}{e} \mid \Tapp{e}{T}  & \text{(Extended Expressions)}  \\[1ex]
   v  & ::= &  \ldots \mid  \Tabs{\alpha\subtpe T}{e} & \text{(Extended Values)}  \\[1ex]   
   p & ::= & x<\tuple{x\tpe T}> & \text{(Typed Join Patterns)}\\[1ex]
   \multicolumn{3}{l}{\alpha, \beta, \gamma \ldots} & \text{(Type-Variables)}
\end{array}
&
\begin{array}[t]{l@{\hspace{1ex}}l@{\hspace{1ex}}l@{\hspace{2em}}l}
   T  & ::= &  \Top \mid \Unit \mid \alpha \mid   \Tsvc{\tuple{T}} \mid \Tsrv{\tuple{ x \tpe T}} \mid \Tsrv{\bot} \mid  \Tsrvinst{T} &   \text{(Types)}   \\
      &     & \phantom{\Top}\mid   \Timg{T} \mid   \Tuni{\alpha \subtpe T}{T} \\[1ex]
 \Gamma & ::= &   \varnothing \mid   \Gamma, \alpha\subtpe T \mid   \Gamma, x \tpe T \mid   \Gamma, \<this> \tpe T & \text{(Type Contexts)}   \\[1ex]
  \Delta & ::= &   \varnothing \mid   \Delta, i \tpe T    & \text{(Location Typings)}
\end{array}
  \end{array}$
  \caption{Expression Syntax of \lang{} with Types}
\label{fig:typeddjcsyntax}
\end{figure*}


\begin{figure*}[t]
\figfontsize
\let\oldtextsc\textsc
\let\oldafterrule\afterruleskip
\renewcommand{\textsc}[1]{{\scriptsize\oldtextsc{#1}}}
\renewcommand{\afterruleskip}{2ex}
\setlength{\columnsep}{20pt}
\typicallabel{\textsc{T-Spwn}} 
\centering
\begin{multicols}{2}
  \infrule[\textsc{T-Var}] {\Gamma(x) = T \andalso x \in
    \Names\cup\{\<this>\}} { \Tjudge{\Gamma\mid\Sigma}{x}{T}}
  \infrule[\textsc{T-Par}] {
    \repeatJdgS{i}{\Tjudge{\Gamma\mid\Sigma}{e_{i}}{\Unit}}}
  {\Tjudge{\Gamma\mid\Sigma}{\prlll{\tuple{e}}}{\Unit} }
\end{multicols}  
  \vspace{\afterruleskip}
  \vspace{\afterruleskip}
  \infrule[\textsc{T-Srv}] { r_i = \tuple{p_i} \triangleright e_i
    \andalso p_{i,j} = x_{i,j}<\tuple{y_{i,j}\tpe T_{i,j}}> \andalso
    S_{i,j} = \Tsvc{\tuple{T_{i,j}}} \vspace{.5ex}\\
    T = \Tsrv{\tuple{x_{i,j}\tpe S_{i,j}}} \andalso
    (\forall\ i,j,k.\ \ j \neq k \to \tuple{y_{i,j}} \cap \tuple{y_{i,k}} = \emptyset) \vspace{.5ex}\\
    (\forall\ i,j,k,l.\ \ x_{i,j} = x_{k,l} \to T_{i,j} = T_{k,l})
    \andalso
    \ftv{T}\subseteq\ftv{\Gamma} \vspace{.5ex}\\
    \repeatJdgS{i}{\Tjudge {\Gamma, \tuple{y_{i,j} \tpe
          T_{i,j}} ,\<this> \tpe T\mid\Sigma} {e_{i}} {\Unit}} } {
    \Tjudge{\Gamma\mid\Sigma}{\srvt{\tuple{r}}}{T}} \hfil
  \infrule[\textsc{T-$\dead$}]{}{\Tjudge{\Gamma\mid\Sigma}{\dead}{\Timg{\Tsrv{\bot}}}}
\vspace{\afterruleskip}
  \infrule[\textsc{T-Img}]{\Tjudge{\Gamma\mid\Sigma}{\srvt{\tuple{r}}}{T}
    \andalso r_i = \tuple{p_i} \triangleright e_i \andalso
    p_{i,j} = x_{i,j}<\tuple{y_{i,j}\tpe T_{i,j}}>  \\
    (\forall k.\exists i,j. (m_{k} = x_{i,j}<\tuple{v_{i,j}}> \wedge
    \tuple{\Tjudge{\Gamma\mid\Sigma}{v_{i,j}}{T_{i,j}}}))
  }{\Tjudge{\Gamma\mid\Sigma}{(\srvt{\tuple{r}},
      \tuple{m})}{\Timg{T}}} 
\vspace{\afterruleskip}
\vspace{\afterruleskip}
\vspace{\afterruleskip}
  \begin{multicols}{2}
    \infrule[\textsc{T-Snap}]{\Tjudge{\Gamma\mid\Sigma}{e}{\Tsrvinst{T}}}{\Tjudge{\Gamma\mid\Sigma}{\<snap>
        e}{\Timg{T}}} 
\vspace{\afterruleskip}
    \infrule[\textsc{T-Repl}]{\Tjudge{\Gamma\mid\Sigma}{e_{1}}{\Tsrvinst{T}}
      \andalso
      \Tjudge{\Gamma\mid\Sigma}{e_{2}}{\Timg{T}}}{\Tjudge{\Gamma\mid\Sigma}{\repl{e_{1}}{e_{2}}}{\Unit}}
\vspace{\afterruleskip}
    \infrule[\textsc{T-Spwn}]
      {\Tjudge{\Gamma\mid\Sigma}{e}{\Timg{T}}}
      {\Tjudge{\Gamma\mid\Sigma}{\<spwn> e}{\Tsrvinst{T}} }
\vspace{\afterruleskip}
      \infrule[\textsc{T-Inst}]
      { 
        \Sigma{(i)} = \Timg{T}} {
        \Tjudge{\Gamma\mid\Sigma}{i}{\Tsrvinst{T}} }
\vspace{\afterruleskip}
      \infrule[\textsc{T-Svc}]
      { 
        \Tjudge{\Gamma\mid\Sigma}{e}{\Tsrvinst{(\Tsrv{\tuple{x\tpe
                T}})}}} {
        \Tjudge{\Gamma\mid\Sigma}{\svc{e}{x_i}}{T_i}}
\vspace{\afterruleskip}
      \infrule[\textsc{T-Req}]
      {\Tjudge{\Gamma\mid\Sigma}{e}{\Tsvc{T_1 \ldots T_n} }\andalso
        \repeatJdgS{i}{\Tjudge{\Gamma\mid\Sigma}{e_{i}}{T_{i}}}}
      {\Tjudge{\Gamma\mid\Sigma}{ e < e_{1}\ldots{}e_{n} > }{\Unit}}
\vspace{\afterruleskip}
        \infrule[\textsc{T-TAbs}]
        { 
          \Tjudge{\Gamma,\alpha\subtpe T\mid\Sigma{}}{e}{U}}
        {\Tjudge{\Gamma\mid\Sigma}{\Tabs{\alpha\subtpe{}T}{e}}{\
            \Tuni{\alpha\subtpe T}{U} }}
\vspace{\afterruleskip}
        \infrule[\textsc{T-TApp}]
        {\Tjudge{\Gamma\mid\Sigma}{e}{\Tuni{\alpha\subtpe{}T_2}{T}} \andalso
          \Subjudge{\Gamma{}}{T_1}{T_2}\\
          \ftv{T_1} \subseteq \ftv{\Gamma}}
        {\Tjudge{\Gamma\mid\Sigma}{\Tapp{e}{T_1}}{T\!\Subst{\alpha\!
              := \!T_1\!}}}
\vspace{\afterruleskip}
      \infrule[\textsc{T-Sub}]
      {\Tjudge{\Gamma\mid\Sigma}{e}{T} \andalso
        \Subjudge{\Gamma}{T}{U}} {\Tjudge{\Gamma{}}{e}{U}}
  \end{multicols}
  \vspace{1em}
  \caption{Typing rules of \lang.}
\label{fig:typerules}
\renewcommand{\textsc}{\oldtextsc}
\renewcommand{\afterruleskip}{\oldafterrule}
\end{figure*}


\begin{figure*}[t]
\figfontsize
\let\oldtextsc\textsc
\let\oldafterrule\afterruleskip
\renewcommand{\textsc}[1]{{\scriptsize\oldtextsc{#1}}}
\renewcommand{\afterruleskip}{2ex}
\setlength{\columnsep}{20pt}
\typicallabel{\textsc{T-Spwn}} 
\centering
\begin{multicols}{2}
  \infrule[\textsc{S-Top}] {} {\Subjudge{\Gamma{}}{T}{\Top}}
\vskip\afterruleskip
  \infrule[\textsc{S-TVar}] {\alpha \subtpe T \in \Gamma{}}
  {\Subjudge{\Gamma{}}{\alpha{}}{T}}
\vskip\afterruleskip
  \infrule[\textsc{S-Srv}] {\repeatJdgS{j}{\existsJdg{i}{(y_j = x_i
        \wedge \Subjudge{\Gamma}{T_i}{U_j})}}} {\Subjudge {\Gamma}
    {\Tsrv{ \tuple{x \tpe T}}} {\Tsrv{ \tuple{y \tpe U}}}}
\vskip\afterruleskip
  \infrule[\textsc{S-Inst}] {\Subjudge {\Gamma} {T} {U}} {\Subjudge
    {\Gamma} {\Tsrvinst{T}} {\Tsrvinst{U}}}
\vskip\afterruleskip
  \infrule[\textsc{S-Srv$_{\bot}$}] {}
  {\Subjudge{\Gamma}{\Tsrv{\bot}}{\Tsrv{T}}}
\vskip\afterruleskip
  \infrule[\textsc{S-Img}] {\Subjudge {\Gamma} {T} {U}} {\Subjudge
    {\Gamma} {\Timg{T}} {\Timg{U}}}
\vskip\afterruleskip
  \infrule[\textsc{S-Svc}]
  {\repeatJdgS{i}{\Subjudge{\Gamma}{U_i}{T_i}}}
  {\Subjudge{\Gamma}{\Tsvc{T_1,\ldots,T_n}}{\Tsvc{U_1,\ldots,U_n}}}
\vskip\afterruleskip
  \infrule[\textsc{S-Univ}] {
    \Subjudge{\Gamma,\alpha_1 \subtpe
      T}{U_1}{U_2\!\Subst{\alpha_2:=\alpha_1}}} {\Subjudge {\Gamma}
    {(\Tuni{\alpha_1 \subtpe T}{U_1})} {(\Tuni{\alpha_2 \subtpe
        T}{U_2})}}
\vskip\afterruleskip
  \infrule[\textsc{S-Refl}]{}{\Subjudge{\Gamma}{T}{T}}
\vskip\afterruleskip
  \infrule[\textsc{S-Trans}] {\Subjudge{\Gamma}{T_1}{T_2} \andalso
    \Subjudge{\Gamma}{T_2}{T_3}} {\Subjudge{\Gamma}{T_1}{T_3}}
\end{multicols}
 \vspace{1em}
  \caption{Subtyping rules of \lang{}.}
  \label{fig:subtypingrules}
\renewcommand{\textsc}{\oldtextsc}
\renewcommand{\afterruleskip}{\oldafterrule}
\end{figure*}


}{}

\noindent {\bf Notation}: In examples, 
\cplinline{'p' & 'p' }
denotes pairs of join patterns and $e \parallel e$ 
denotes pairs of parallel expressions. We sometimes omit empty buffers when spawning servers, i.e.,
we write $\<spwn> (\srvt{\tuple{r}})$ for $\<spwn> (\srvt{\tuple{r}}, \varepsilon)$.
To improve readability in larger examples, we use curly braces to indicate the lexical scope
of syntactic forms. We write service names and meta-level definitions in typewriter 
font, e.g., \cplinline{this#foo} and \cplinline{MyServer = srv \{ \}}. We write bound variables
in italic font, e.g., \cplinline{$\;$srv$\;$\{ left<'x'> & right<'y'> :> pair<'x', 'y'> \}}.

\paragraph{Example.} 
For illustration,
consider the
following server template \cplinline{Fact} for computing factorials, which defines three rules with 5 services.\footnote{For the sake of
presentation, we use ordinary notation for numbers, arithmetics and conditionals,
all of which is church-encodable on top of \lang{} (cf. Section~\ref{sec:derived-syntax}).}
\begin{lstlisting}[language=cpl]
Fact = srv {
  main<'n', 'k'> :> //initialization
    this#fac<'n'> || this#acc<1> || this#out<'k'>

  fac<'n'> & acc<'a'> :> //recursive fac computation
    if ('n' <= 1)
    then this#res<'a'>
    else (this#fac<'n' - 1> || this#acc<'a' * 'n'>)

  res<'n'> & out<'k'> :> 'k'<'n'> //send result to k
}
\end{lstlisting}

%
%
The first
rule defines a service \cplinline{main} with two arguments, an integer $n$ and a
continuation $k$. The continuation is
necessary because service requests are asynchronous and thus, the
factorial server must notify the caller when the computation finishes. Upon
receiving a \cplinline{main} request, the server sends itself three requests: \cplinline{fac}
represents the outstanding factorial computation, \cplinline{acc} is used as an
accumulator for the ongoing computation, and \cplinline{out} stores the
continuation provided by the caller.

The second rule of \cplinline{Fact} implements the factorial function and
synchronously matches and consumes requests \cplinline{fac} and \cplinline{acc} 
using join patterns. Upon termination,
the second rule sends a request \cplinline{res} to the running server instance,
otherwise it decreases the argument of \cplinline{fac} and updates the accumulator.  Finally, the
third rule of \cplinline{Fact} retrieves the user-provided continuation $k$ from the
request \cplinline{out} and the result \cplinline{res}. The rule expects the continuation to be a service reference
and sends a request to it with the final result as argument.

To compute a factorial, we create a server instance from the template
\cplinline{Fact} and request service \cplinline{main}: 
\begin{lstlisting}[numbers=none,language=cpl,mathescape=true,aboveskip=1ex,belowskip=1ex]
(spwn Fact)#main<5, 'k'>$.$
\end{lstlisting}

\ifdef{\tecreport}{An example reduction trace is in the appendix.}{}

\subsection{Operational Semantics}
\label{sec:oper-semant}

We define the semantics of \lang{} as a small-step structural
operational semantics using reduction contexts $E$ (Figure~\ref{fig:djcsyntax}) in the style of Felleisen and
Hieb~\cite{FelleisenH92}.

Figure~\ref{fig:djcsemantics} shows the reduction rules for \lang{} expressions. 
Reduction steps are atomic and take the form $e\,{\mid}\,\mu --> e'\,{\mid}\,\mu'$.
A pair $e\,{\mid}\,\mu$ represents a distributed cloud application, where expression $e$ describes
its current behavior and $\mu$ describes its current \emph{distributed} state. 
We intend $e$ as a description of the software components and resources that execute and reside 
at the cloud provider and do not model client devices.
We call the component $\mu$ a \emph{routing table}, which is a finite map.
Intuitively, $\mu$ records which addresses a cloud provider assigns to the server instances 
that the cloud application creates during its execution.\footnote{This bears similarity to lambda calculi enriched with references and a store~\cite{Wright199438}.}  
We abstract over technical details, such as the underlying network.

The first reduction rule \inflabel{Cong} defines the congruence rules of the
language and is standard.
The second rule \inflabel{Par} is technical. It flattens nested parallel expressions in order
to have a simpler representation of parallel computations.
The third rule \inflabel{Rcv} lets a server instance receive an
asynchronous service request, where the request is added to the instance's buffer
for later processing. Our semantics abstracts over the technicalities of network communication. 
That is, we consider requests $\svc{i}{x}<\tuple{v}>$ that occur in a \lang{} expression
to be in transit, until a corresponding \inflabel{Rcv} step consumes them.
The fourth rule \inflabel{React} fires reaction rules of a server. 
It selects a running server instance $(s, \tuple{m})$,
selects a reaction rule $(\tuple{p} \triangleright e_b)$ from it, 
and tries to match its join patterns $\tuple{p}$ against the pending service requests
in the buffer $\tuple{m}$. A successful match consumes the service requests,
instantiates the body $e_b$ of the selected reaction rule and executes it
independently in parallel.

Finally, let us consider the rules for $\<spwn>\!\!$, $\<snap>\!\!$ and $\<repl>\!\!$, which manage 
server instances and images. 
Reduction rule \inflabel{Spwn} creates a new server 
instance from a server image, 
where a fresh unique address is assigned to the server instance.
This is the only rule that allocates new addresses in $\mu$. One can think of
this rule as a request to the cloud provider to create a new virtual machine
and return its IP address.
Importantly, the address that $\<spwn>\!\!$ yields is only visible to the caller.
The address can only be accessed by another expression if it shares a common
lexical scope with the caller. Thus, lexical scope restricts the visibility
of addresses. This also means that the map $\mu$ is not a shared memory, but
a combined, flat view of disjoint distributed information.\footnote{This approach is comparable to sets of definitions in the chemical soup of the Join Calculus~\cite{Fournet-Gonthier:popl96}.}

Reduction rule \inflabel{Snap} yields a copy of the server image at address $i$,
provided the address is in use. Intuitively, it represents the invocation of a cloud 
management API to create a virtual machine snapshot. 
Reduction rule \inflabel{Repl} replaces the server image at address $i$ with another server image.

We define $\<spwn>\!\!$, $\<snap>\!\!$ and $\<repl>\!\!$ as atomic operations. 
At the implementation level, each operation may involve multiple communication steps with the cloud provider, 
taking noticeable time to complete and thus block execution for too long, especially when the operation translates to booting a new OS-level
virtual machine. On the other hand, as we motivated at the beginning of this section, servers
may not necessarily map to virtual machines, but in-memory computations. In this case, 
we expect our three atomic operations to be reasonably fast.
Also, we do not impose any synchronization mechanism on a server addresses, which may 
result in data races if multiple management operations access it in parallel. Instead,
programmers have to write their own synchronization mechanisms on top of \lang{} if required.

\ifdef{\tecreport}{
\noindent
Matching satisfies the following property:
\begin{proposition}[Match soundness and completeness]\label{thm:mtchsndnss}
  Let $\tuple{p}$ be a sequence of join patterns with
  $p_i = x_i<\tuple{y_i}>$, $\tuple{m}$ and $\tuple{m_r}$ sequences of
  service request values, and $\sigma$ a substitution.
  $\matchp{\tuple{p},\tuple{m}}{\tuple{m_r},\sigma}$ if and only if it exists a
  sequence $\tuple{m_c}$ such that:
\begin{enumerate}
\item\label{mtchsndnss1} Sequence $\tuple{m_c}$ represents the requests values consumed
  from \tuple{m}, that is, $\tuple{m_r}\ \tuple{m_c} = \tuple{m}$ modulo
  permutation.
\item\label{mtchsndnss2} All consumed requests $\tuple{m_c}$ 
  match the join patterns $\tuple{p}$, that is, $\tuple{m_c}$ and $\tuple{p}$
  have the same length and $m_{c,i} = x_i<\tuple{v_i}>$, where
  $\tuple{y_i}$ and $\tuple{v_i}$ have the same length.
\item\label{mtchsndnss3} $\sigma$ substitutes the parameters of the matched join
  patterns with the actual arguments, that is,
  $$\sigma = \Subst{\tuple{y_i} := \tuple{v_i} \mid 1 \le i \le k }$$ where $k$ is the length of $\tuple{p}$.
\end{enumerate}
\end{proposition}
  \begin{proof}
    Soundness ($\Rightarrow$): Straightforward induction on the derivation of the judgment
    $\matchp{\tuple{p},\tuple{m}_{1}}{\tuple{m}_{2}, \sigma}$. Completeness ($\Leftarrow$):
    Straightforward by induction on the number $k$ of service patterns in $p$.
 \end{proof}
}{}

Our semantics is nondeterminstic along 3 dimensions:
\begin{itemize}
\item If multiple server instances can fire a rule, \inflabel{React} selects one
  of them nondeterminstically. This  models concurrent
  execution of servers that can react to incoming service requests independently.
\item If multiple rules of a server instance can fire, \inflabel{React} selects one of
  them nondeterminstically. This is of lesser importance and languages building
  on ours may fix a specific order for firing rules (e.g., in the
  order of definition).
\item If multiple service request values can satisfy a join pattern,
  \inflabel{Match$_1$} selects one of them nondeterminstically. This 
  models asynchronous communication in distributed systems, i.e., the
  order in which a server serves requests is independent of the order in
  which services are requested. More concrete languages based on \lang{} may employ
  stricter ordering (e.g., to preserve the order of requests that originate from a
  single server).
\end{itemize}

\subsection{Placement of Servers.}
\label{sec:locality}
We intentionally designed the semantics of \lang{} with transparency
of server \emph{placement} in mind. That is, a single abstraction in the language, the
server instance, models all computations, irrespective of whether the instance runs on its own physical
machine or as a virtual machine hosted remotely -- indeed, placement transparency
is a distinguishing feature of cloud applications. 

However, despite the behavior of servers being invariant to placement,
placement has a significant impact in real-world scenarios and
influences communication and computation
performance~\cite{BobroffKB07,MengPZ10}. The need to account for
placement in an implementation is critical considering that -- servers
being the only supported abstraction -- every single let binding and
lambda abstraction desugars to a server spawn
(cf.~Section~\ref{sec:derived-syntax}). In our concurrent Scala implementation, we support an extended syntax for server
spawns that allows programmers to declare whether a server instance runs
in a new thread or in the thread that executes the spawn. This provides a simple
mechanism for manually implementing placement strategies. 

A viable alternative to manual specification of placement are automatic
placement strategies. Together with server migration, automatic placement
strategies can adapt the server layout to changing conditions. Based on our language, a management system
for a cloud infrastructure can formally reason about optimal placement strategies. 
In future work, we plan to implement these ideas in a distributed run-time system for \lang{} (cf. Section \ref{sec:interf-with-cloud}).

\ifdef{\tecreport}{
\subsection{Derived syntax and base operations}
\label{sec:derived-syntax}

Our core language is expressive enough to encode a wide range of typical
language constructs. To illustrate its expressiveness and for
convenience in expressing example computations in the rest of the paper, we
define derived syntax for let-expressions, first-class functions, thunks, and
base operations, all of which can be desugared to the core syntax introduced above.

\paragraph{Let bindings.} The derived syntax for \cplinline{let}
bindings desugars to the core syntax of the \lang{} as follows:
$$\<let> x = e_{1} \<in> e_{2}\ \  \rightsquigarrow\ \  \svc{(\<spwn>(\srvt{ \mathtt{let}<x> \triangleright e_{2} }))}{\mathtt{let}}<e_{1}>.   $$
Evaluating \cplinline{let} amounts to (a) spawning a new server instance that offers a
service called $\mathtt{let}$ that will run $e_2$ when requested and (b) requesting
this service with the bound expression $e_{1}$ as an argument.

We also define derived syntax for a variant of \cplinline{let} called \cplinline{letk} for
cases in which the bound expression provides its result through a
continuation. This is to account for the fact that often expressions in the
\lang{} involve asynchronous service calls that, instead of returning a
value, pass it to a continuation.
The definition of \cplinline{letk} is as follows:
$$\<letk> x = e_{1}<\tuple{e}> \<in> e_{2} \rightsquigarrow e_{1}<\tuple{e}, \svc{(\<spwn>(\srvt{ \mathtt{k}<x> \triangleright e_{2} }))}{\mathtt{k}}>.   $$
Here, we bind the variable $x$ via continuation that we add to the service
request $e_{1}<\tuple{e}>$, assuming $e_1$ takes a continuation as final
argument. When $e_1$ terminates, it calls the continuation and thus triggers
execution of $e_2$. 
For example, we can use \cplinline{letk} to bind and use the result of the \cplinline{Fact}
server shown above:
$$ \<letk> \id{n} = \svc{(\<spwn> \mathtt{Fact})}{\mathtt{main}}<\mathtt{5}> \<in> \svc{\mathtt{Log}}{\mathtt{write}}<n> $$

\noindent Note that the desugaring for both variants of \cplinline{let} wrap the body $e_{2}$ in
a server template, which changes the meaning of the self reference \cplinline{this} in
$e_{2}$. To counter this effect and to make the derived syntax transparent, the
desugaring that we actually implemented substitutes free occurrences of
\cplinline{this} in $e_2$ to the server instance surrounding the \cplinline{let}.

\paragraph{First-class functions.} We can encode first-class functions as
server instances with a single service \cplinline{app}:
$$\lambda x.\; e \rightsquigarrow \<spwn> (\srvt{\mathtt{app}<x, k> \triangleright \CPS{e}{k}}), \ \  \text{where $k$ is fresh}$$
Recall that service requests in \lang{} are asynchronous. In order to
correctly propagate argument values and the result of function bodies, we need
to transform argument expressions and function bodies into continuation-passing
style, for example using the following transformation \CPSname:
\\[1ex]
$\figfontsize\begin{array}{l@{\ }l@{\ }l}
  \CPS{\lambda x.\; e}{k} &=& k<\<spwn> (\srvt{\mathtt{app}<x, k> \triangleright \CPS{e}{k}})> \\
  \CPS{(f\ e)}{k} &=& \CPS{f}{\svc{(\<spwn> (\srvt{\mathtt{k_1}<v_f> \triangleright \hskip2em\text{where $v_f$ is fresh} \\
        && \quad \CPS{e}{\svc{(\<spwn> (\srvt{\mathtt{k_2}<v_e> \triangleright \hfill\text{where $v_e$ is fresh} \\
              && \quad\quad \svc{v_f}{\mathtt{app}}<v_e, k>}))}{\mathtt{k_2}}}}))}{\mathtt{k_1}}} \\
  \CPS{e}{k} &=& k<e> \\
\end{array}$
\\[1ex]
For example, we can define and apply a function that instantiates a server-template argument:
$$\svc{(\lambda x.\; \<spwn> x)}{\mathtt{app}}<\mathtt{Fact}, k_0> -->*  k_0<\mathtt{Fact^0}>$$

\noindent Our encoding of first-class functions is similar to the one proposed
for the Join Calculus~\cite{Fournet-Gonthier:popl96} and it also shows that our
language is Turing-complete. Moreover, it enables Church-encodings of data types
such as numbers or lists.

\paragraph{Thunks.} A thunk is a first-class value that represents a packaged,
delayed computation. Servers can force the computation of a thunk and they can
pass thunks to other servers. Thunks play a significant role in distributed
systems, because they enable servers to distribute work over other servers
dynamically.

Interestingly, lambdas as defined above do not give rise to a useful
implementation of thunks, because a computation that is encoded as a lambda is
already installed on a concrete spawned server: Every lambda expression gives
rise to exactly one server instance that solely executes the body of this
lambda. In contrast, we desire an implementation of thunks that allows us to
dynamically allocate servers for executing a thunk.  To this end, we represent
thunks as server templates:
$$\<thunk>\; e \rightsquigarrow \srvt{\mathtt{force}<k> \triangleright k<e>}$$
Since server templates are first-class in \lang{}, thunks can be passed
to other servers. A server can instantiate a thunk any number of
times and call the \cplinline{force} request with a continuation to get the result of
the thunk.

Note that similarly to \cplinline{let}, we substitute \cplinline{this} in thunks and 
lambda abstractions by the enclosing server instance to make our encodings transparent.

\paragraph{Base operations.} While we can use Church encodings to represent data
types and their operations, it is more convenient (and more efficient in
practice) to assume some built-in base operations. In particular, we can take
the liberty of assuming that base operations are synchronous and in direct
style, that is, base operations return values directly and do not require
continuation-passing style. For the remainder of the paper, we assume built-in
base operations on Booleans, integers, floating points, tuples and
lists. We added these operations in our implementation and it is easy to add
further base operations.  To distinguish synchronous calls to base operations
from asynchronous service requests, we use rounded parentheses for base
operations, for example, $\functionsymbol{max}{7,11}$.


}{
\input{derived_syntax_short}
}


\ifdef{\tecreport}{
\section{Type System}
\label{sec:type-system}

We designed and formalized a type system for \lang{} in the style of
System F with subtyping and bounded quantification~\cite{pierce2002types}. The type system ensures that all service requests in a
well-typed program refer to valid service declarations with the correct number
of arguments and the right argument types. Subtyping enables us to define
public server interfaces, where the actual server implementation defines private
services, for example, to manage internal state.

Figure \ref{fig:typeddjcsyntax} shows the syntax of types, typing contexts, location
typings as well as extensions to expressions and values. Similar to lambda calculi
with references, alongside standard typing contexts $\Gamma$ we also
record the type of server instances at each allocated address via location typings $\Sigma$.
A type $T$ is either the top type
$\Top$, the unit type $\Unit$, a type variable $\alpha$, a service type
$\Tsvc{\tuple{T}}$ representing a service with arguments of type $T_i$, a
server-template type $\Tsrv{\tuple{x\tpe T}}$ of a server with services $x_i$ of
type $T_i$, the special server-template type $\Tsrv{\bot}$ for inactive servers, a server-instance type $\Tsrvinst{T}$, 
a server-image type $\Timg{T}$ or a universal type
$\Tuni{\alpha \subtpe T_1}{T_2}$.
The syntax of typing contexts and the
extensions of expressions and values is standard: We have type abstraction 
$\Tabs{\alpha\subtpe T}{e}$, and type application $\Tapp{e}{T}$. Finally, we require type annotations in join patterns.

We define the typing judgment $\Tjudge{\Gamma\mid\Sigma}{e}{T}$ by the rules depicted in
figure~\ref{fig:typerules}. The rules are mostly
straightforward. \inflabel{T-Var} looks up the type of a variable or $\<this>$
in the context. \inflabel{T-Par} requires that all expressions of a parallel
expression have type \Unit.

\inflabel{T-Srv} is the most complicated type rule. Intuitively, the type of a
server template is the set of all services that the server offers. $r_i$
represents rule number $i$ of the server template, where $p_{i,j}$ is pattern
number $j$ of rule number $i$. Patterns $p_{i,j}$ provide services $x_{i,j}$,
which have service type $S_{i,j}$. The type $T$ of the server template then
consists of all provided services with their types. To make sure the server
template is well-typed, we check that join patterns are linear (service
parameters are distinct), services in different patterns have
consistent types and that all free type variables are bound ($\ftv{T}\subseteq \ftv{\Gamma}$). Finally, we
check the right-hand side $e_i$ of each reaction rule, where we bind all service
parameters $y_{i,j}$ as well as $\<this>$. 

Next, we define three introduction rules for server image types.
The first is \inflabel{T-$\dead$}, which specifies that $\dead$ is an image of an
inert server. The second rule  \inflabel{T-Img} types server image values $(\srvt{\tuple{r}}, \tuple{m})$,
where we require that $\srvt{\tuple{r}}$ is a well-typed server template and
 each service request value in the buffer $\tuple{m}$ is understood
by this server template. That is, each value $m_{k}$ in $\tuple{m}$
must correspond to a join pattern mentioned in $\tuple{r}$ and the arguments must have the types
which are annotated in the join pattern.
The last introduction rule for server image types is $\inflabel{T-Snap}$ for snapshots, which requires
that the argument to $\<snap>\!\!$ is actually a server instance in order to yield a corresponding server image.

Rule \inflabel{T-Repl} types replacements $\repl{e_{1}}{e_{2}}$ as $\Unit$ and requires that
replacements are preserving the interface of the server instance to be replaced.
That is, the first argument must be an instance with interface type $T$ and the second argument 
an image type for the same interface type. 

There are two introduction rules for server instances.
\inflabel{T-Spwn} requires the argument of $\<spwn>\!\!$ to be a server
image in order to yield a corresponding instance. 
Rule \inflabel{T-Inst} handles server addresses, which must be allocated
in the location typing $\Sigma$ to a server image. \todo{Double check the entire paper to get the wording on server image, instance and address right }

\inflabel{T-Svc} defines
that a service reference is well-typed if the queried server provides a service
of the required name. \inflabel{T-Req} requires that the target of a service
request indeed is a service reference and that the request has the right number of
arguments with the right types. The remaining four type rules are standard.

Figure~\ref{fig:subtypingrules} defines the subtyping relation
$\Subjudge{\Gamma}{T}{T}$. We employ width subtyping and depth subtyping for
server-template types, that is, the server subtype can provide more services
than the server supertype and the server subtype can provide specialized
versions of services promised by the server supertype. 
A special case is rule \inflabel{S-Srv$_{\bot}$}, which specifies that
the type $\Tsrv{\bot}$ for inert servers is a subtype of every other server template type.
This ensures that $\dead$ can be placed in every context requiring an image of type $\Timg{\Tsrv{T}}$.
The other subtyping rules are straightforward.

\paragraph{Preservation.}
\label{sec:type-preservation}

We prove preservation for our type system using standard substitution
lemmas~\cite{pierce2002types}. The proofs appear in the appendix at the end of the paper. 

\begin{lemma}[Substitution Lemma]
If $\Tjudge{\Gamma,x \tpe  T_1\mid\Sigma}{e_2}{T_2}$ and $\Tjudge{\Gamma\mid\Sigma}{e_1}{T_1}$
then $\Tjudge{\Gamma\mid\Sigma}{\nobreak e_{2}\!\Subst{x := e_1}}{T_2}$.
\end{lemma}

\begin{lemma}[Type Substitution Preserves Subtyping]\ \\
  If $\Subjudge{\Gamma, \alpha\subtpe T', \Gamma'}{S}{T}$ and $\Subjudge{\Gamma}{S'}{T'}$
  then \\ $\Subjudge{\Gamma, \Gamma'\sigma}{S\sigma}{T\sigma}$ where 
  $\sigma = \Subst{\alpha := S'}$.
\end{lemma}

\begin{lemma}[Type Substitution Lemma]
   If $\Tjudge{\Gamma, \alpha\subtpe S,\Gamma'\mid\Sigma}{e}{T}$ and $\Subjudge{\Gamma}{S'}{S}$ then $\Tjudge{\Gamma, \Gamma'\sigma\mid\Sigma\sigma}{e\sigma}{T\sigma}$
   where $\sigma = \Subst{\alpha := S'}$.
\end{lemma}

\begin{lemma}[Location Typing Extension Preserves Types]\ \\
  If $\Tjudge{\Gamma\mid\Sigma}{e}{T}$ and $\Sigma\subseteq\Sigma'$, then $\Tjudge{\Gamma\mid\Sigma'}{e}{T}$.
\end{lemma}

 \begin{lemma}[Replacement]
   If $\mathcal{D}$ is a derivation with root \Tjudge{\Gamma}{\context{E}{e}}{U}, $\mathcal{D}' \subdereq \mathcal{D} $
a derivation with root \Tjudge{\Gamma'}{e}{U'} and \Tjudge{\Gamma'}{e'}{U'}, then $\Tjudge{\Gamma}{\context{E}{e'}}{U}$.
 \end{lemma}

\begin{theorem}[Preservation]\ \\
   If $\Tjudge{\Gamma\mid\Sigma}{e}{T}$ and  $\Gamma\mid\Sigma\vdash\mu$ and $e\mid\nobreak\mu --> e'\mid\nobreak\mu'$, then $\Tjudge{\Gamma\mid\Sigma'}{e'}{T}$ for some $\Sigma'$, where $\Sigma\subseteq\Sigma'$ and $\Gamma\mid\Sigma'\vdash \mu'$.
\end{theorem}

\noindent Note that the proof of the preservation theorem requires the match soundness property
from proposition~\ref{thm:mtchsndnss}, in order to verify that after the reduction step
of rule $\inflabel{React}$ (fig.~\ref{fig:djcsyntax}), consumption of service requests
and instantiation of rule bodies preserves the type of the enclosing parallel composition.

\paragraph{Progress.}
\label{sec:progress}

Our type system does not satisfy progress. For example, the following program is
well-typed and not a value but cannot reduce:
\\[1ex]$
\begin{array}{ll}
\svc{(\<spwn> (\srvt{\mathtt{foo}<> \joined \mathtt{bar}<> \triangleright  \prlll{\varepsilon}}))}{\mathtt{foo}}<>.
 \end{array}
$\\[1ex]
The service request \cplinline{foo} resolves fine, but the server's rule cannot fire
because it is lacking a request \cplinline{bar} joined with \cplinline{foo}. Since our type
system does not correlate service requests, it cannot guarantee that join
patterns must succeed eventually.
The integration of such a property 
is an interesting
direction of future work, but orthogonal to the main contributions of this work.

\paragraph{Auxiliary Notation}
\label{sec:notation-types}

We adopt the following conventions.
We omit the type bound if it is $\Top$, e.g.,
 $\Tabs{\alpha\subtpe\Top}{}$ becomes  $\Tabs{\alpha}{e}$.
It is sometimes convenient to abbreviate longer type expressions.
We introduce an abbreviation syntax for faux type constructors, e.g,
$\Type{SVC}[\alpha_{1}, \ldots \alpha_{n}] := <\alpha_{1}, \Tsvc{\alpha_{2}, \ldots, \alpha_{n}}>$
defines a type $\Type{SVC}$  with $n$
free type variables. Writing $\Type{SVC}[T_{1}, \ldots T_{n}]$
denotes the type obtained by substituting
the free occurrences of $\alpha_{1}, \ldots, \alpha_{n}$ with the provided types $T_{1}, \ldots T_{n}$.
Instead of repeating type annotations of service parameters in every
join pattern, we declare the service types once at the beginning of server templates.
For example, $\srvt{( \mathtt{a}<x\tpe\Type{Int}> \trr \mathtt{foo}<>)\ \ (\mathtt{a}<x\tpe\Type{Int}> \joined \mathtt{b}<y\tpe\Tsvc{\Type{Bool}}> \trr \mathtt{bar}<>     )}$
becomes 
\begin{align*}
        \srvt{&\mathtt{a}\tpe <\Type{Int}>, \mathtt{b}\tpe < <\Type{Bool} > >\\
        & (\mathtt{a}<x>\trr \mathtt{foo}<>)\ \ (\mathtt{a}<x>\joined \mathtt{b}<x> \trr \mathtt{bar}<> )}.
\end{align*}
We define function types as $$\Tfun{T_{1},\ldots T_{n}}{T} := <T_{1},\ldots, T_{n}, <T> >$$
following our function encoding in Section \ref{sec:derived-syntax}.
We define the union $(\Tsrv{\tuple{x\tpe T}}) \cup
(\Tsrv{\tuple{y\tpe U}})$ of two server-template types as the server template that
contains all services of both types. The union is only defined if service names
that occur in both types have identical type annotations -- they are merely for syntactic convenience and do not represent real union
types, which we leave for future work.

}{\input{typesystem_short}}

\section{CPL at Work}
\label{sec:case-studies}

We present two case studies to demonstrate the adequacy of \lang{} for solving the deployment  
issues identified in Section~\ref{sec:motivation}. 
The case studies will also be subsequently used to answer research questions about \lang{}'s features.

Firstly, we developed a number of reusable {\it server combinators}, expressing deployment  patterns found in cloud computing. Our examples focus on load balancing and fault tolerance, demonstrating that programmers
can define their own cloud services as strongly-typed, composable 
modules and address nonfunctional requirements with \lang{}.
Secondly, we use our language to model MapReduce~\cite{Lammel20081} deployments for distributed batch computations.
Finally, we apply our server combinators to MapReduce, effortlessly obtaining a type-safe composition of services.

\subsection{Server Combinators}
\label{sec:server-combinators}

In Section~\ref{sec:motivation}, we identified extensibility issues
with deployment languages, which prevents programmers from integrating
their own service implementations. We show how to implement
custom service functionality with \emph{server combinators} in a type-safe and
composable way. Our combinators are similar in spirit to higher-order functions
in functional programming.

As the basis for our combinators, we introduce {\bf workers}, i.e.,
servers providing computational resources. A worker accepts work packages as thunks. 
Concretely, a worker models a managed virtual machine in a cloud and thunks
model application services.

Following our derived syntax for thunks (Section~\ref{sec:derived-syntax}),
given an expression $e$ of type $\alpha$, the type of $\<thunk> e$ is:
\begin{lstlisting}[language=cpl,numbers=none,mathescape=true,aboveskip=1ex,belowskip=1ex]
             !TThunk![%a] $:=$ srv force:<%a>$.$
\end{lstlisting}
Service \cplinline{force} accepts a continuation and calls it with the result of
evaluating $e$. A worker accepts a thunk and executes it. At the type level,
workers are values of a polymorphic type
\begin{lstlisting}[language=cpl,numbers=none,mathescape=true,aboveskip=1ex,belowskip=1ex]
 !TWorker![$\alpha$] $:=$ srv init:<>, work:!TThunk![$\alpha$]$\;$->$\;\alpha$$.$
\end{lstlisting}
That is, to execute a thunk on a worker, clients request the \cplinline{work} service
which maps the thunk to a result value. In
addition, we allow workers to provide initialization logic via a service
\cplinline{init}. Clients of a worker should request \cplinline{init} before they issue
\cplinline{work} requests. Figure~\ref{fig:worker_factory}  defines a factory for creating 
basic workers, which have no initialization logic and execute thunks in their own instance scope. 
In the following, we define server combinators that enrich
workers with more advanced features.    
\begin{figure}[t]
  \centering
\begin{cplcode}
MkWorker[%a] = srv {
 make: () -> !TWorker![%a]
 make<'k'> :>
  let 'worker' = spwn srv {
   init: <>, work: !TThunk![%a] -> %a  
   init<> :> par %e //stub, do nothing
   work<'thnk', 'k'> :> (spwn 'thnk')#force<'k'>
  } in 'k'<'worker'>
}
\end{cplcode}
  \caption{Basic worker factory.}
  \label{fig:worker_factory}
\end{figure}


To model {\bf locality} -- a worker uses its own computational
resources to execute thunks -- the spawn of a thunk should in fact \emph{not}
yield a new remote server instance. As discussed in Section~\ref{sec:locality}, 
to keep the core language minimal the operational semantics does not 
distinguish whether a server is local or remote to another server. 
 However, in our concurrent
implementation of \lang, we allow users to annotate spawns as being
remote or local, which enables us to model worker-local execution of thunks. 

The combinators follow a common design principle. 
{\it (i)} The combinator is a factory for
server templates, which is a server instance with a single 
\cplinline{make} service. The service accepts one or more server
templates which implement the \Type{TWorker} interface, among possibly other arguments.
{\it (ii)} Our combinators produce
\emph{proxy} workers. That is, the resulting workers
implement the worker interface but forward requests of the
$\id{work}$ service to an internal instance of the argument worker.

\subsubsection{Load Balancing}
\label{sec:load-balancing}

A common feature of cloud computing is on-demand scalability of services by
dynamically acquiring server instances and distributing load among them.
\lang{} supports the encoding of on-demand scalability in form of a server combinator,
that distributes load over multiple workers dynamically, given a user-defined decision algorithm.

Dynamically distributing load requires a means to approximate worker utilization. 
Our first combinator \linebreak\cplinline{MkLoadAware} enriches workers with the ability to
answer \cplinline{getLoad} requests, which sends the current number of pending requests
of the \cplinline{work} service, our measure for utilization. Therefore, the corresponding type\footnote{The union $\cup$ on server types is for notational convenience at the meta level and \emph{not} part of the type language.}
for load aware workers is
\begin{lstlisting}[language=cpl,numbers=none,mathescape=true,aboveskip=1ex,belowskip=1ex]
 !TLAWorker![$\alpha$] $:=$ !TWorker![%a] $\cup$ srv getLoad:<<!Int!>>$.$
\end{lstlisting}
\begin{figure}[t]
  \centering
\begin{lstlisting}[language=cpl,mathescape=true]
MkLoadAware[%a, %w <: !TWorker![%a]] = srv {
 make: %w -> !TLAWorker![%a]
 make<'worker', 'k'> :>
  let 'lWorker' = srv {
   instnc: <inst %w>, getLoad: () -> !Int!, load: <!Int!> 
   work: !TThunk![%a] -> %a, init: <>   
   //... initialization logic omitted 
   
   //forwarding logic for work
   work<'thnk', 'k'> & instnc<'w'> & load<n> :> (*@\label{lst:mkloadaware:fwdstart}@*)
    this#load<'n'+1> || this#instnc<'w'> 
    || letk 'res' = 'w'#work<'thnk'> 
       in ('k'<'res'>$\;$||$\;$this#done<>) (*@\label{lst:mkloadaware:fwdend}@*)

   //callback logic for fullfilled requests
   done<> & load<'n'> :> this#load<'n'-1> (*@\label{lst:mkloadaware:callback}@*)
   
   getLoad<'k'> & load<'n'> :> 'k'<'n'> || this#load<'n'>
  } in 'k'<'lWorker'>
}
\end{lstlisting}
  \caption{Combinator for producing load-aware workers.}
  \label{fig:mkloadaware}
\vspace{-0.8em}
\end{figure}


The \cplinline{make} service of the combinator accepts a server template \cplinline{'worker'} implementing the \Type{TWorker} interface and
returns its enhanced version (bound to \cplinline{'lWorker'}) back on the given continuation \cplinline{'k'}.
Lines \ref{lst:mkloadaware:fwdstart}-\ref{lst:mkloadaware:fwdend} implement the core idea of forwarding and counting the pending requests.
Continuation passing style enables us to intercept and hook on to the responses of \cplinline{'worker'} after finishing work requests, 
which we express in Line~\ref{lst:mkloadaware:fwdend} by the \cplinline{letk} construct. 

By building upon load-aware workers, we can define a polymorphic combinator \cplinline{MkBalanced} that
transparently introduces load balancing over a list of load-aware workers. The combinator is flexible in that
it abstracts over the scheduling algorithm, which is an impure polymorphic function of type
\begin{lstlisting}[language=cpl,numbers=none,mathescape=true,belowskip=1ex,aboveskip=1ex]
!Choose![%w] $:=$ !List![inst %w] -> !Pair![inst %w, !List![inst %w]]$.$
\end{lstlisting}
Given a (church-encoded) list of possible worker instances, such a function returns a (church-encoded) pair consisting
of the chosen worker and an updated list of workers, allowing for dynamic adjustment of the available worker pool (\emph{elastic load balancing}).

Figure~\ref{fig:mkbalanced} shows the full definition of the \cplinline{MkBalanced} combinator.
\begin{figure}[t]
  \centering
\begin{lstlisting}[language=cpl,mathescape=true]
MkBalanced[%a, %w <: !TLAWorker![%a]] = srv {
 make: (!List![%w], !Choose![%w]) -> !TWorker![%a]
 make<'workers', 'choose', 'k'> :>
  let 'lbWorker' = srv {
   insts: <!List![inst %w]>, 
   work: !TThunk![%a] -> %a, init: <>
   
   init<> :> //spawn and init all child workers
    letk 'spawned' = mapk<'workers', %l'w':%w.$\;$spwn$\;$'w'>  
     in (this#insts<'spawned'> 
         ||foreach<'spawned', %l'inst':inst$\;$%w.$\;$'inst'#init<>>)

   //forward to the next child worker
   work<'thnk', 'k'> & insts<'l'> :> 
    letk ('w', $l'$) = 'choose'<'l'> (*@\label{lst:mkbalanced:fwdstart}@*)
     in ('w'#work<'thnk', 'k'> || this#insts<$l'$>) (*@\label{lst:mkbalanced:fwdend}@*)
  } in 'k'<'lbWorker'>
}
\end{lstlisting}
  \caption{Combinator for producing load-balanced workers.}
  \label{fig:mkbalanced}
\end{figure}
%

%
Similarly to Figure~\ref{fig:mkloadaware}, the combinator is a factory which produces a decorated worker. The only
difference being that now there is a list of possible workers to forward requests to.
Choosing a worker is just a matter of querying the scheduling algorithm \cplinline{'choose'} (Lines \ref{lst:mkbalanced:fwdstart}-\ref{lst:mkbalanced:fwdend}).
Note that this combinator is only applicable to server templates implementing the $\Type{TLAWorker}[\alpha]$ interface (Line 1), 
since \cplinline{'choose'} should be able to base its decision on the current load of the workers. 

In summary, mapping a list of workers with \cplinline{MkLoadAware} and passing the
result to \cplinline{MkBalanced} yields a composite, load-balancing worker.
It is thus easy to define hierarchies of load balancers programmatically by 
repeated use of the two combinators. Continuation passing style and the type system 
enable flexible, type-safe compositions of workers.

\subsubsection{Failure Recovery}
\label{sec:failure-recovery}

Cloud platforms monitor virtual machine instances to ensure their continual
availability. We model failure recovery for crash/omission, permanent, fail-silent
failures~\cite{Tanenbaum:2006:DSP:1202502}, where a failure makes a
virtual machine unresponsive and is recovered by a restart.

Following the same design principles of the previous section, we
can define a failure recovery combinator \linebreak$\cplinline{MkRecover}$, that produces fault-tolerant workers.
\ifdef{\tecreport}{Its definition is in the appendix of this report.}{We omit its definition and refer to our technical report~\cite{techreport}.}

Self-recovering workers follow a basic protocol. Each time a $\cplinline{work}$ request
is processed, we store the given thunk and continuation in a list
until the underlying worker confirms the request's completion. If the wait time exceeds a
timeout, we replace the worker with a fresh new instance
and replay all pending requests.
Crucial to this combinator is the \cplinline{repl} syntactic form, which swaps the
running server instance at the worker's address: \cplinline{repl$\;$'w'$\;$('worker',$\;\varepsilon$)}.
Resetting the worker's state amounts to setting the empty buffer $\varepsilon$ in the
server image value we pass to \cplinline{repl}.


\ifdef{\tecreport}
{
\subsection{MapReduce}
\label{sec:cloud-calculus-at}

In this section, we illustrate how to implement the MapReduce~\cite{Dean:2008:MSD:1327452.1327492} programming model with typed combinators in \lang,
taking fault tolerance and load balancing into account. 
MapReduce facilitates parallel data processing -- cloud platforms are a desirable deployment target.
The main point we want to make with this
example is that \lang{} programs do not exhibit the unsafe composition, non-extensibility and staging problems
we found in Section~\ref{sec:motivation}.
Our design is inspired by Lämmel's formal presentation in Haskell~\cite{Lammel20081}.

Figure~\ref{fig:mapreduce_factory} shows the main combinator for creating a MapReduce deployment,
which is a first-class server.
\begin{figure}[t]
  \centering
\begin{lstlisting}[language=cpl,mathescape=true]
MapReduce[%k$_{1}$,%v$_{1}$,%k$_{2}$,%v$_{2}$,%v$_{3}$] = spwn srv {
 make: (!TMap![%k$_{1}$,%v$_{1}$,%k$_{2}$,%v$_{2}$], 
        !TReduce![%k$_{2}$,%v$_{2}$,%v$_{3}$],
        !TPartition![%k$_{2}$], 
        $\forall\alpha.$() -> !TWorker![%a],
        Int) -> !TMR![%k$_{1}$,%v$_{1}$,%k$_{2}$,%v$_{3}$]

 make<'Map', 'Reduce', 'Partition', 'R', 'mkWorker', 'k'> :> (*@\label{lst:mapreduce:make}@*)
  let 'sv' = srv { (*@\label{lst:mapreduce:appstart}@*)
   app<'data', 'k'$_{0}$> :> let  (*@\label{lst:mapreduce:append}@*)
    'mworker' = (*@\label{lst:mapreduce:mapallocstart}@*)
      mapValues('data', %l'v'.$\;$'mkWorker'[List[Pair[%k$_{1}$,$\;$%v$_{2}$]]])
    'rworker' =
      mkMap(map(range(1, 'R'), %l'i'. ('i', 'mkWorker'[%v$_{3}$]))) (*@\label{lst:mapreduce:mapallocend}@*)
    'grouper' =
      MkGrouper<'Partition', 'R', 'Reduce'
                size('mworker'), 'rworker', 'k'$_{0}$>
   in foreach<'data', %l'key', 'val'. { (*@\label{lst:mapreduce:foreachstart}@*)
    let 'thnk' = thunk 'Map'<'key', 'val'> (*@\label{lst:mapreduce:thunk}@*)
     in$\;$get('mworker', 'key')#work<'thnk', 'grouper'#group>}>   (*@\label{lst:mapreduce:foreachend}@*)
  } in 'k'<'sv'>
}     
\end{lstlisting}
  \caption{MapReduce factory.}
  \label{fig:mapreduce_factory}
\end{figure}

Following Lämmel's presentation, the combinator is generic in the key and value types.
$\kappa_{i}$  denotes type parameters for keys and $\nu_{i}$ denotes type parameters for values.

The combinator takes as parameters the \id{Map} function for decomposing input key-value pairs into
an intermediate list of intermediate pairs, the \id{Reduce} function for transforming grouped
intermediate values into a final result value, the \id{Partition} function which controls grouping
and distribution among reducers and the number \id{R} of reducers to allocate (Line~\ref{lst:mapreduce:make}).
Parameter \id{mkWorker} is a polymorphic factory of type \cplinline{$\forall\alpha.$() -> !TWorker![$\alpha$]}.
It produces worker instances for both the map and reduce stage.

Invoking \cplinline{make} creates a new server template that on invocation of its \cplinline{app} service
deploys and executes a distributed MapReduce computation for a given set of (church-encoded) key-value pairs $\id{data}$
and returns the result on continuation $k_{0}$ (Lines \ref{lst:mapreduce:appstart}-\ref{lst:mapreduce:append}).

Firstly, workers for mapping and reducing are allocated and stored in the local map data structures $\id{mworker}$ and $\id{rworker}$,
where we assume appropriate cps-encoded functions that create and transform maps and sequences (Lines \ref{lst:mapreduce:mapallocstart}-\ref{lst:mapreduce:mapallocend}).
Each key in the input $\id{data}$ is assigned a new mapping worker 
and each partition from $1$ to $\id{R}$ is assigned a reducing worker. 
Additionally, a component for grouping and distribution among reducers (\id{grouper}) is allocated. 

Secondly, the \cplinline{foreach} invocation (Lines \ref{lst:mapreduce:foreachstart}-\ref{lst:mapreduce:foreachend}) distributes key-value pairs in parallel among mapping workers.
For each pair, the corresponding worker should invoke the \id{Map} function, which we express as a thunk (Line~\ref{lst:mapreduce:thunk}, cf.\ section~\ref{sec:derived-syntax}).
All resulting intermediate values are forwarded to the grouper's \cplinline{group} service. 

\begin{figure}[t]
  \centering
\begin{lstlisting}[language=cpl,mathescape=true]
MkGrouper[%k$_{2}$,%v$_{2}$,%v$_{3}$] = spwn srv {
 make<'Partition': (%k$_{2}$, Int) -> Int, 'R': Int,
      'Reduce': !TReduce![%k$_{2}$,%v$_{2}$,%v$_{3}$],
      'R': Int,
      'rworker': Map[Int, !TWorker![%v$_{3}$]],
      'k'$_{r}$: <Pair[%k$_{2}$,%v$_{3}$]>,
      'k': <inst srv (group: <List[Pair[%k$_{2}$,%v$_{2}$]]>)> :>
 let 'grpr' = spwn srv { 
  //accumulated per-partition values for reduce
  state: <Map[Int, Map[%k$_{2}$,%v$_{2}$]]>,
  
  //result callback invoked by mappers
  group: <List[Pair[%k$_{2}$,%v$_{2}$]]>,
  
  //waiting state for phase one
  await: <Int>, 

  //trigger for phase two 
  done: <Map[Int, Map[%k$_{2}$,%v$_{2}$]]>

  //phase one: wait for mapper results
  state<'m'> & group<'kvs'> & await<'n'> :> (*@\label{lst:grouper:awaitstart}@*)
   letk 'm'$'$ = foldk<'kvs', 'm', 
         %l'm'$''$, 'kv'. letk 'i' = 'Partition'<fst('kv'), 'R'>
                  in updateGroup('m'$''$, 'i', 'kv')>
   in if ('n' $\gt$ 0)
      then (this#await<'n' - 1> || this#state<'m'$'$>)
      else this#done<'m'$'$> (*@\label{lst:grouper:awaitend}@*)

  //phase two: distribute to reducers
  done<'m'> :> (*@\label{lst:grouper:reducestart}@*)
   foreach<'m', %l'i'. 'data2'.
    foreach<'data2', %l'key','vals'. 
     let 'thnk' = thunk 'Reduce'<'key', 'vals'>
     in get('rworker', 'i')#work<'thnk', 'k'$_{r}$>>> (*@\label{lst:grouper:reduceend}@*)
  } in 'grpr'#state<emptyMap> || 'grpr'#await<'R'> || 'k'<'grpr'>
}
\end{lstlisting}
  \caption{Grouper factory.}
  \label{fig:mapreduce_grouper_factory}
\end{figure}

The grouper (Figure~\ref{fig:mapreduce_grouper_factory}) consolidates 
multiple intermediate values with the same key and forwards them to the reducer workers.
It operates in phases: (1) wait for all mapper workers to finish, meanwhile grouping incoming results (Lines~\ref{lst:grouper:awaitstart}-\ref{lst:grouper:awaitend}) and (2),
assign grouped results to reducer workers with the \id{Partition} function and distribute as thunks, 
which invoke the \id{Reduce} function (Lines~\ref{lst:grouper:reducestart}-\ref{lst:grouper:reduceend}).
All reduction results are forwarded to the continuation $k_{r}$. For brevity we omit the final merge of the results.

Thanks to our service combinators, we can easily address non-functional requirements 
and non-intrusively add new features. The choice of the \id{mkWorker} parameter 
determines which variants of MapReduce deployments we obtain: The default variant just
employs workers without advanced features, i.e.,
\begin{lstlisting}[numbers=none,language=cpl,mathescape=true]
let 'make' = %L%a.(spwn MkWorker[%a])#make
in MapReduce[%k$_{1}$,%v$_{1}$,%k$_{2}$,%v$_{2}$,%v$_{3}$]#make<'f',$\;$'r',$\;$'p',$\;$'R',$\;$'make',$\;$'k'>
\end{lstlisting}

for appropriate choices of the other MapReduce parameters.

In order to obtain a variant, where worker nodes are elastically load-balanced, 
one replaces \cplinline{'make'} with \cplinline{'makeLB'} below, which composes the combinators from the previous section:
\begin{lstlisting}[numbers=none,language=cpl,mathescape=true,aboveskip=.5ex,belowskip=.5ex]
let$\;$'choose' = ...//load balancing algorithm
   $\;$'makeLB' = %L%a.%l'k'. {
    letk$\;$'w' = (spwn MkWorker[%a])#make<>
     $\;$'lw' = (spwn$\;$MkLoadAware[%a,$\;$!TWorker![%a]])#make<'w'>     
     in$\;$(spwn MkBalanced[%a, !TLAWorker![%a]])
            #make<mkList('lw'),$\;$'choose',$\;$'k'> 
    }
in 'makeLB'
\end{lstlisting}

A similar composition with the fault tolerance combinator yields fault tolerant
MapReduce, where crashed mapper and reducer workers are automatically recovered.

}
{\input{mapreduce_short}}

\subsection{Discussion}
\label{sec:case-study-discussion}

We discuss how \lang performed in the case studies answering the following research questions:
\begin{itemize}
\item[Q1]{\sf (Safety)}: {\it Does \lang{} improve safety of cloud deployments?}

\item[Q2] {\sf (Extensibility)}: {\it Does \lang{} enable custom and extensible service implementations?}

\item[Q3] {\sf (Dynamic self-adjustment)}: {\it Does \lang{} improve
    flexibility in dynamic reconfiguration of deployments ?}

\end{itemize}

\paragraph{Safety}
\lang{} is a strongly-typed language. 
As such, it provides internal safety (Section~\ref{sec:motivation}). 
The issue of cross-language safety (Section~\ref{sec:motivation}) does not occur in \lang{} programs, because 
configuration and deployment code are part of the same application.
In addition, the interconnection of components is well-typed. 
For example, in the MapReduce case study, it is guaranteed that worker invocations 
cannot go wrong due to wrongly typed arguments. 
It is also guaranteed that workers yield values of the required types.
 As a result, all mapper and reducer workers are guaranteed to be compatible with the grouper component. 
In a traditional deployment program, interconnecting components amounts to referring to each others attributes, 
but due to the plain syntactic expansion, there is no guarantee of compatibility.

\paragraph{Extensibility} The possibility to define combinators in \lang{}
supports extensible, custom service implementations. At the type system level, 
bounded polymorphism and subtyping ensure that service implementations 
implement the required interfaces. The load balancing example enables 
nested load balancing trees, since the combinator implements the 
well-known Composite design pattern  from object-oriented programming.
At the operational level, continuation passing style enables flexible
composition of components, e.g., for stacking multiple features.


\paragraph{Dynamic Self-Adjustment}

In the case studies, we encountered the need of dynamically adapting the deployment configuration of an application,
which is also known as ``elasticity''.
For example, the load balancer combinator can easily support dynamic growth or shrinkage of the list of available workers: 
New workers need to be dynamically deployed in new VMs (growth) and certain VMs must be halted and removed from the 
cloud configuration when the respective workers are not needed (shrinkage).
Dynamic reconfiguration is not directly expressible in configuration languages, due to the two-phase staging.
For example, configurations can refer to external elastic load balancer services provided
by the cloud platform, but such services only provide a fixed set of balancing strategies,
which may not suit the application. The load balancer service can be regarded as a black box, which
happens to implement  elasticity features.
Also, a configuration language can request load balancing services 
only to the fixed set of machines which is specified in a configuration, but it is not possible
if the number of machines is unknown before execution, as in the MapReduce case study.
In contrast, \lang users can specify their own load balancing strategies and apply them programmatically.

\subsection{Interfacing with Cloud Platforms}
\label{sec:interf-with-cloud}

A practical implementation of \lang requires (1) a mapping of its
concepts to real-world cloud platforms and (2) integrate existing
cloud APIs and middleware services written in other languages.  In the
following, we sketch a viable solution; we leave a detailed implementation
for future work.

For (1), \lang programs can be compiled to bytecode and be interpreted by a distributed 
run time hosted on multiple virtual machines. 

Concerning (2), we envision our structural server types as the interface of \lang{}'s run time with the external world, i.e.,
pre-existing cloud services and artifacts written in other languages. 
\lang developers must write wrapper libraries to implement typed language bindings.
Indeed, \lang{}'s first-class servers resemble (remote) objects, where services are their methods 
and requests are asynchronous method invocations (returning results on continuations). 
\lang{} implementations hence can learn from work
on language bindings in existing object-oriented language run times, e.g., the Java ecosystem.
To ensure type safety, dynamic type checking is necessary at the boundary between
our run time and components written in dynamically or weakly typed languages.

Note that the representation of external services and artifacts as servers requires \emph{immutable}
addresses. That is, the run time should forbid \cplinline{snap} and \cplinline{repl} on such objects, because it is in general impossible to reify a state snapshot of the external world.

For the primitives \cplinline{spwn}, \cplinline{snap},
and \cplinline{repl}, the run time must be able to orchestrate the virtualization facilities of the cloud provider
via APIs.
Following our annotation-based
approach to placement (Section \ref{sec:locality}), these primitives
either map to local objects or to fresh virtual machines.
Thus, invoking \cplinline{spwn$\;$'v'} may create a new virtual machine hosting the 
\lang{} run time, which allocates and runs $v$. For
local servers, $v$ is executed by the run time that invoked \cplinline{spwn}.
One could extend the primitives to allow greater control of infrastructure-level concerns, such as
machine configuration and geographic distribution.
From these requirements and \lang{}'s design targeting extensible services and distributed applications, 
it follows that \lang{} is cross-cutting the three abstraction layers in contemporary cloud 
platforms: Infrastructure as a Service (IaaS), Platform as a Service (PaaS) and Software as a Service (SaaS)~\cite{Vaquero:2008:BCT:1496091.1496100}.





\section{Related Work}\label{sec:related}

\paragraph{Programming Models for Cloud Computing.}

The popularity of cloud computing
infrastructures~\cite{Vaquero:2008:BCT:1496091.1496100} has encouraged
the investigation of programming models that can benefit from
on-demand, scalable computational power and feature location
transparency. Examples of these languages, often employed in the
context of big data analysis, are Dryad~\cite{Isard:2009:DDC:1559845.1559962},
PigLatin~\cite{Olston:2008:PLN:1376616.1376726} and
FlumeJava~\cite{Chambers:2010:FEE:1806596.1806638}. These languages
are motivated by refinements and generalizations of the original
MapReduce~\cite{Dean:2008:MSD:1327452.1327492} model. 

Unlike \lang{}, these models specifically target only certain kinds of cloud 
computations, i.e., massive parallel computations and derivations thereof. 
They deliberately restrict the programming model to enable 
automated deployment, and do not address deployment programmability in 
the same language setting as \lang{} does. In this paper, we showed that 
the server abstraction of \lang{} can perfectly well model MapReduce 
computations in a highly parametric way, but it covers at the same time a 
more generic application programming model as well as deployment 
programmability. Especially, due to its join-based synchronization, \lang{}
is well suited to serve as a core language for modeling cloud-managed 
stream processing. 

Some researchers have investigated by means of formal methods specific
computational models or specific aspects of cloud computing. The
foundations in functional programming of MapReduce have been studied
by Lämmel~\cite{Lammel20081}. In \lang it is possible to
encode higher-order functions and hence we can model MapReduce's
functionality running on a cloud computing platform. Jarraya et
al.~\cite{6261089} extend the Ambient calculus to account for firewall
rules and permissions to verify security properties of cloud
platforms. To the best of our knowledge, no attempts have been done in
formalizing cloud infrastructures in their generality.

\paragraph{Formal Calculi for Concurrent and Distributed Services.} 

Milner's CCS~\cite{Milner:1982:CCS:539036}, the $\pi$
calculus~\cite{Milner19921} and Hoare's
CSP~\cite{Hoare:1978:CSP:359576.359585} have been studied as the
foundation of parallel execution and process synchronization.

Fournet's and Gonthier's Join Calculus~\cite{Fournet-Gonthier:popl96}
introduced join patterns for expressing the interaction 
among a set of processes that communicate by asynchronous message passing
over communication channels. 
The model of communication channels in this calculus 
more adequately reflects communication primitives in real world
computing systems which allows for a simpler implementation.
In contrast, the notion of channel in the previously mentioned
process calculi would require expensive global consensus protocols
in implementations. 

The design of \lang{} borrows join patterns from the Join Calculus. 
Channels in the Join Calculus are similar 
to services in \lang, but the Join Calculus does not have first-class and 
higher-order servers with qualified names. Also, there is no support for 
deployment abstractions.

The Ambient calculus~\cite{Cardelli2000177} has been developed by
Cardelli and Gordon to model concurrent systems that include both
mobile devices and mobile computation. Ambients are a notion of named, bounded places
where computations occur and can be moved as a whole to other places. Nested
ambients model administrative domains and capabilities control access to ambients. 
\lang, in contrast, is location-transparent, which is
faithful to the abstraction of a singular entity offered by cloud applications.

\paragraph{Languages for Parallel Execution and Process Synchronization.} Several languages have been successfully
developed/extended to support features studied in formal calculi. 

JoCaml is an ML-like implementation of Join Calculus which adopts
state machines to efficiently support join
patterns~\cite{LeFessant1998205}. Polyphonic
C\#~\cite{Benton-Cardelli-Fournet:2004} extends C\# with join-like
concurrency abstractions for asynchronous programming that are
compiler-checked and optimized. Scala
Joins~\cite{Haller-VanCutsem:coordination2008} uses Scala's
extensible pattern matching to express joins. The Join Concurrency
Library~\cite{Russo:2007ux} is a more portable implementation of
Polyphonic C\# features by using C\# 2.0
generics. JEScala~\cite{VanHam:2014:JMC:2577080.2577082} combines
concurrency abstraction in the style of the Join Calculus with
implicit invocation.

Funnel~\cite{AnOverviewofFunct:2002uz} uses the Join Calculus as its
foundations and supports object-oriented programming with classes and
inheritance. Finally, JErlang~\cite{Plociniczak:2010ed} extends the
Erlang actor-based concurrency model. Channels are messages exchanged
by actors, and received patterns are extended to express matching of
multiple subsequent messages. Turon and Russo~\cite{Turon-Russo:oopsla2011} propose an efficient, lock-free implementation of the join matching algorithm 
demonstrating that  declarative specifications with joins can scale to complex coordination problems with good performance --
even outperforming specialized algorithms. 
Fournet et al.~\cite{Fournet:2000:ADI:647318.723475} provide an
implementation of the Ambient calculus. The implementation is obtained
through a formally-proved translation to JoCaml.

\lang{} shares some features with these languages, basically those 
built on the Join Calculus. In principle, the discussion about the relation 
of \lang{} to Join Calculus applies to these languages as well, since the Join Calculus is 
their shared foundation. Implementations of \lang can benefit from the techniques
developed in this class of works, especially \cite{Russo:2007ux}.


\ifdef{\tecreport}{}{\vspace{-1ex}}
\section{Conclusions and Future Work}
\label{sec:conclusion}
We presented \lang, a statically typed core language for defining asynchronous 
cloud services and their deployment on cloud platforms. \lang{}
improves over the state of the art for cloud deployment DSLs:
It enables (1) statically safe service composition, (2) custom implementations
of cloud services that are composable and extensible and (3) dynamic
changes to a deployed application.
In future work, we will implement and expand core \lang to a practical programming language for cloud applications and deployment.


\ifdef{\tecreport}{
\acks This work has been supported by the European Research Council, grant No. 321217.
}
{\acks We thank the anonymous reviewers for their helpful comments. This work has been supported by the European Research Council, grant No. 321217.}
\bibliographystyle{abbrvnat} 
\bibliography{report} 

\ifdef{\tecreport}{
\newpage

\appendix

\section{Example Reduction}
\label{sec:example-reduction}

To illustrate the small-step operational semantics of the
\lang{}, we investigate the reduction trace of a service request for
computing the factorial of $3$ using the server template \cplinline{Fact} defined
above. Here, we assume $k_0$ is some continuation interested in the result of
the computation. We write $\varnothing$ to denote the empty routing table
and  $(\mathtt{Fact}^n, \tuple{m})$ to denote a routing table entry at address $n$ for a server instance of the template \cplinline{Fact} with buffer \tuple{m}.
Further we write $n\textnormal{/}{r_i}$ to refer
to rule number $i$ of the server at $n$. For brevity, reductions of the rule \inflabel{Par}, as well as 
reductions of if-then-else and
arithmetic expressions are omitted. Multiple subsequent reduction steps of a rule \inflabel{R} are denoted by $\inflabel{R}^{*}$.
\[
\figfontsize
\begin{array}{p{1.8cm}p{2.5cm}l}
     \multicolumn{2}{l}{\svc{(\<spwn> \mathtt{Fact})}{\mathtt{main}}<3, k_0>} & \mid\varnothing \\[1ex]
     $\xrightarrow{\textsc{Spwn}}$ & $\svc{0}{\mathtt{main}} <3, k_0>$ &\mid \Set{(\mathtt{Fact}^{0}, \varepsilon)} \\[1ex]
$\xrightarrow{\textsc{Rcv}}\ $ & $\prlll{\varepsilon}$ &\mid \Set{(\mathtt{Fact}^{0}, \mathtt{main}<3, k_{0}>} \\[1ex]
$\xrightarrow{\textsc{React}\ 0\textnormal{/}{r_1}}\ $  & $\svc{0}{\mathtt{fac}}< 3 >$   &\mid\Set{(\mathtt{Fact}^{0}, \varepsilon)} \\
& $\parallel \svc{0}{\mathtt{acc}}<1>$ & \\
& $\parallel \svc{0}{\mathtt{out}}<k_0>$ & \\[1ex]
$\xrightarrow{\textsc{Rcv}^{*}}\ $  & $\prlll{\varepsilon}$  &\mid\{(\mathtt{Fact}^{0}, \mathtt{fac}<3>\; \mathtt{acc}<1>\\
& & \quad\quad\quad\quad\ \ \;\mathtt{out}<k_{0})\} \\[1ex]
$\xrightarrow{\textsc{Rcv}^{*}, \textsc{React}\ 0\textnormal{/}{r_2}}\ $  & $\svc{0}{\mathtt{fac}}<2> \parallel \svc{0}{\mathtt{acc}}<3>$  &\mid \Set{(\mathtt{Fact}^{0}, \mathtt{out}<k_{0}>)}  \\[1ex]
$\xrightarrow{\textsc{Rcv}^{*},\textsc{React}\ 0\textnormal{/}{r_2}}\ $  & $\svc{0}{\mathtt{fac}}<1> \parallel \svc{0}{\mathtt{acc}}<6>$  &\mid \Set{(\mathtt{Fact}^{0}, \mathtt{out}<k_{0}>)} \\[1ex]
$\xrightarrow{\textsc{Rcv}^{*},\textsc{React}\ 0\textnormal{/}{r_2}}\ $ & $\svc{0}{\mathtt{res}}<6>$  &\mid\Set{(\mathtt{Fact}^{0}, \mathtt{out}<k_{0}>)} \\[1ex]
$\xrightarrow{\textsc{Rcv},\textsc{React}\ 0\textnormal{/}{r_3}}\ $ & $k_0<6>$  &\mid \Set{(\mathtt{Fact}^{0}, \varepsilon)}.
\end{array}
\]
The reduction starts with the instantiation of the server template
\cplinline{Fact}. Next, we serve service \cplinline{main} on the server instance
$\mathtt{Fact}^0$, which yields $3$ service requests to $\mathtt{Fact}^0$. We can fire
the second rule of $\mathtt{Fact}^0$ three times by matching services \cplinline{fac} and
\cplinline{acc}. In fact, the second rule of $\mathtt{Fact}^0$ is the only rule that can
fire. The first two times, the argument of \cplinline{fac} is larger than $1$, so we
execute the else branch, which yields updated values for \cplinline{fac} and \cplinline{acc}. The
third time, the argument of \cplinline{fac} is $1$, so we execute the then branch,
which yields a service request \cplinline{res} with the final result. Finally, we execute
the third rule of $\mathtt{Fact}^0$, which matches services \cplinline{res} and \cplinline{out} 
to forward the final result to the continuation $k_0$.

We can spawn multiple instances of \cplinline{Fact} and compute factorials in parallel:
$$\svc{(\<spwn> \mathtt{Fact})}{\mathtt{main}}<3, k_0> \parallel \svc{(\<spwn> \mathtt{Fact})}{\mathtt{main}}<5, k_1>$$
Due to the nondeterminism of our semantics, some of the possible reduction
traces interleave computations of both factorials. However, 
since requests always contain the target address and rule \textsc{React}
operates solely on a server instance's buffer, there cannot be any interference
between two different instances of \cplinline{Fact}.
This way, e.g., each instance of \cplinline{Fact} has its own accumulator.

\newpage

\section{Case Studies}
\label{sec:appendix:case-studies-1}

In the following, we give the full definition of the server combinators and actor supervision case studies,
which we omitted in section~\ref{sec:case-studies} due to space limitations.

\subsection{Server Combinators for Cloud Computing}
\label{sec:appendix-c}

\subsubsection{Failure Recovery}
\label{sec:appendix:failure-recovery-1}
The combinator for failure recovery:
\begin{lstlisting}[language=cpl,mathescape=true]
MkRecover[%a,%w$\;$<:!TWorker![%a]] = spwn srv {
 make: (%w, !Int!) -> !TWorker![%a]

 make<'worker', 'timeout', 'k'> :>
  let 'self-recovering' = srv {
   init: <>, 
   work: !TThunk![%a] -> %a, 
   instnc: <inst %w>
   pending: <!List![(!Int!, !Int!, !TThunk![%a], <%a>)]>, 
   done: <!Int!> 

   //initialization
   init<> :> let 'w' = (spwn 'worker') in
    'w'#init<> || this#inst<'w'> || this#pending<'Nil'>

   //store and forward work requests to instnc
   work<'thnk', 'k'> & instnc<'w'> & pending<'xs'> :>
    let 'ID' = freshID() in
     let 'now' = localTime() in
      this#pending<('ID', 'now', 'thnk', 'k') :: 'xs'>
      || this#instnc<'w'> 
      || (letk 'r' = 'w'#work<'thnk'> 
         $\;$in ('k'<'r'> || this#done<'ID'>))

   //work completion by instnc
   done<'ID'> & pending<'xs'> :>
    filterk('xs', $\lambda$'p'.$\;$fst('p') $\neq$ 'ID', this#pending>)

   //check for timeouts, restart instnc if needed
   pending<'xs'> & instnc<'w'> :>
    let 'now' = localTime() in
     letk 'late' = exists('xs', 
                      $\lambda$'p'.$\,$'now' - snd('p')$\,\gt\,$'timeout') in
      if ('late') then
       (repl 'w' ('worker',%e);$\,$('w'#init<> || this#instnc<'w'>))
       || foreach<xs, $\lambda$'p'.$\;$this#work<thrd(p), frth(p)>>
      else 
       (this#pending<'xs'> || this#instnc<'w'>)
  } in 'k'<'self-recovering'>
}
\end{lstlisting}
Service \cplinline{make} accepts a stoppable worker and an integer timeout
parameter. The first rule of the self-recovering worker initializes the list of
pending requests to the empty list \id{Nil}. The second rule accepts \cplinline{work}
requests. It generates a fresh ID for the request and adds the request to the
list of pending requests together with the ID, the local timestamp and the
continuation (we assume functions \cplinline{freshID} and \cplinline{localTime}). The rule
forwards the work request to the underlying worker and installs a continuation
that notifies the proxy that the request completed using service \cplinline{done}. The
third rule accepts this request and removes the corresponding request from the
list of pending requests.

Finally, the last rule checks if any of the pending requests has a timeout. If
this happens, the rule replaces the old worker instance by a new one via $\<repl>\!\!$, effectively
resetting the state of the worker and re-initializing it.
In parallel, all pending requests are replayed.

\section{Type System Proofs}
\label{sec:append-d:-type}

\begin{definition}
The typed language extends evaluation contexts with type applications:
$$ \context{E}{} ::= \ldots \mid \Tapp{\context{E}{}}{T}.  $$
The reduction relation is extended by an additional contraction rule:
 \infax[TAppAbs]{ \Tapp{(\Tabs{\alpha\subtpe U}{e})}{T} --> e\Subst{\alpha := T}.}
\ 
\end{definition}

\begin{definition}
  We write $\Sigma\subseteq\Sigma'$ if for all $(i\tpe T) \in \Sigma$,
  $(i\tpe T)\in\Sigma'$ holds.
\end{definition}

\begin{definition}\label{def:well-typed-rt}
  A routing table $\mu$ is \emph{well typed} with respect to $\Gamma$, $\Sigma$ (written $\Gamma\mid\Sigma\vdash\mu$),
  if $\dom{\mu} = \dom{\Sigma}$ and for all $i\in\dom{\mu}$, $\Tjudge{\Gamma\mid\Sigma}{\mu(i)}{\Sigma(i)}$ holds.
\end{definition}

\textbf{Note.} In the proofs we use the standard variable convention. That is, 
bound variables are assumed to be distinct and can be renamed if necessary
so that no variable capture can occur in substitutions.

\begin{lemma}[Substitution Lemma]\label{lem:substition}
If $\Tjudge{\Gamma,x \tpe  T_1\mid\Sigma}{e_2}{T_2}$ and $\Tjudge{\Gamma\mid\Sigma}{e_1}{T_1}$
then $\Tjudge{\Gamma\mid\Sigma}{\nobreak e_{2}\!\Subst{x := e_1}}{T_2}$.
\end{lemma}
\begin{proof}
  By induction on the typing derivation $\mathcal{D}$ of $\Tjudge{\Gamma,x \tpe  T_1\mid\Sigma}{e_2}{T_2}$. In each
  case we assume $\Tjudge{\Gamma\mid\Sigma}{e_{1}}{T_{1}}$.
  \newline
  
  \noindent \textbf{Basis:}

  \begin{description}
  \item[\inflabel{T-Var}:] Therefore $e_{2} = y$ for $y\in\Names\cup\Set{\<this>}$ and $(\Gamma, x\tpe T_{1})(y) = T_{2}$.
     Case distinction:
     \begin{description}
     \item[$x = y$:]
       Therefore $e_{2} = x$, $T_{2} = T_{1}$ and $e_{2}\Subst{x := e_{1}} = e_{1}$.
       From this and $\Tjudge{\Gamma\mid\Sigma}{e_{1}}{T_{1}}$  we obtain a derivation of $\Tjudge{\Gamma\mid\Sigma}{ e_{2}\Subst{x := e_{1}}}{T_{2}}$.
     \item[$x \neq y$:]
       Therefore $e_{2}\Subst{x := e_{1}} = y\Subst{x := e_{1}} = y = e_{2}$, hence $\Tjudge{\Gamma\mid\Sigma}{e_{2}\Subst{x := e_{1}}}{T_{2}}$,
       since the assumption $x\tpe T_{1}$ can be dropped.
     \end{description}
     
   \item[\inflabel{T-Inst}:] 
     Immediate, since the context $\Gamma$ is not considered in the premise.
   
   \item[\inflabel{T-$\dead$}:] 
     Immediate.
  \end{description}
  
  \noindent \textbf{Inductive step:}

  \noindent \textit{Induction hypothesis (IH)}: The property holds for all proper subderivations of the derivation $\mathcal{D}$ of $\Tjudge{\Gamma,x \tpe  T_1\mid\Sigma}{e_2}{T_2}$.

  \begin{description}
  \item[\inflabel{T-Par}:] From the conclusion of the rule it holds that $e_{2} = \<par> \tuple{e'_{2}}$, $T_{2} = \Unit$
    and from its premises $\Tjudge{\Gamma,x\tpe T_{1}\mid\Sigma}{e'_{2,i}}{T_{2}}$ for each $e'_{2,i}$ in the sequence $\tuple{e'_{2}}$.
    Applying (IH) to each of the derivations in the premise yields $\Tjudge{\Gamma\mid\Sigma}{e'_{2,i}\Subst{x := e_{1}}}{T_{2}}$ for each $i$.
    Together with rule \inflabel{T-Par} we obtain a derivation for $\Tjudge{\Gamma\mid\Sigma}{\<par> \tuple{e'_{2}\Subst{x := e_{1}}}}{T_{2}}$,
    which is also a derivation for $\Tjudge{\Gamma\mid\Sigma}{e_{2}\Subst{x := e_{1}}}{T_{2}}$ as desired, since 
    $\<par> \tuple{e'_{2}\Subst{x := e_{1}}} = (\<par> \tuple{e'_{2}})\Subst{x := e_{1}} = e_{2}\Subst{x:=e_{1}}$.

  \item[\inflabel{T-Srv}:]  It holds that $e_{2} = \srvt{\tuple{r}}$, $r_{i} = \tuple{p}_{i}\trr e'_{i}$, $T_{2} = \Tsrv{\tuple{x_{i,j}\tpe S_{i,j}}}$,
     $\ftv{T_{2}} \subseteq \ftv{\Gamma, x : T_{1}}$,  
     \linebreak $p_{i,j} =\allowbreak x_{i,j}<\tuple{y_{i,j}:T_{i,j}} >$, $S_{i,j} = <\tuple{T_{i,j}}>$ and\linebreak $\Tjudge{\Gamma, x\tpe T_{1}, \tuple{y_{i,j}\tpe T_{i,j}}, \<this>\tpe T\mid\Sigma}{e'_{i}}{\Unit}$
     for each $r_{i}$ in $\tuple{r}$.
     Note that $\ftv{T_{2}}\subseteq\ftv{\Gamma}$, since $\ftv{\Gamma, x : T_{1}} = \ftv{\Gamma}$ by definition of $\ftv{}$.

     Case distinction:
     \begin{description}
     \item[$x = \<this>$:]
       From the derivations of $\Tjudge{\Gamma, x\tpe T_{1}, \tuple{y_{i,j}\tpe T_{i,j}}, \linebreak\<this>\tpe T\mid\Sigma}{e'_{i}}{\Unit}$ we
       obtain derivations for $\Tjudge{\Gamma, \tuple{y_{i,j}\tpe T_{i,j}}, \<this>\tpe T\mid\Sigma}{e'_{i}}{\Unit}$ since the
       assumption $\<this> : T$ shadows $x : T_{1}$. Since server templates bind $\<this>$,
       it follows that $e_{2}\Subst{x := e_{1}} = e_{2}$. Together with the other assumptions from the original derivation we 
       obtain a derivation for $\Tjudge{\Gamma\mid\Sigma}{e_{2}\Subst{x := e_{1}}}{T_{2}}$ with rule \inflabel{T-Srv} as desired.

     \item[$v\neq\<this>$:] 
       For each $r_{i}$ in $\tuple{r}$ it holds that $x$ distinct from $\tuple{y}_{i,j}$ by the variable convention.
       From the derivation of
        $\Tjudge{\Gamma, x\tpe T_{1}, \tuple{y_{i,j}\tpe T_{i,j}}, \<this>\tpe T\mid\Sigma}{e'_{i}}{\Unit}$
        we obtain by permutation a derivation of
        $\Tjudge{\Gamma, \tuple{y_{i,j}\tpe T_{i,j}},\allowbreak \<this>\tpe T,  x\tpe T_{1}\mid\Sigma}{e'_{i}}{\Unit}$. 
        With (IH) we obtain a derivation for \linebreak$\Tjudge{\Gamma, \tuple{y_{i,j}\tpe T_{i,j}},\allowbreak \<this>\tpe T\mid\Sigma}{e'_{i}\Subst{x := e_{1}}}{\Unit}$.

              From these intermediate derivations and the assumptions from the original \inflabel{T-Srv} derivation we
       obtain by \inflabel{T-Srv} a derivation of
 $\Tjudge{\Gamma\mid\Sigma}{\Tsrv{(\tuple{p}_{i}\triangleright  e'_{i}\Subst{x := e_{1}})}}{T_{2}}$,
       which is also a derivation of $\Tjudge{\Gamma\mid\Sigma}{e_{2}\Subst{x := e_{1}}}{T_{2}}$ as desired.
     \end{description}

  \item[\inflabel{T-Spwn}:] Therefore $e_{2} = \<spwn> e'_{2}$, $T_{2} = \Tsrvinst{T'_{2}}$,  and $\Tjudge{\Gamma, x\tpe T_{1}\mid\Sigma}{e'_{2}}{\Timg{T'_{2}}}$.
   Applying (IH) yields $\Tjudge{\Gamma\mid\Sigma}{e'_{2}\Subst{x := e_{1}}}{\Timg{T'_{2}}}$. Together with rule \inflabel{T-Spwn} we obtain a derivation of
   $\Tjudge{\Gamma\mid\Sigma}{\<spwn> (e'_{2}\Subst{x := e_{1}}) }{T_{2}}$, which is also a derivation of $\Tjudge{\Gamma\mid\Sigma}{e_{2}\Subst{x := e_{1}}}{T_{2}}$ as desired,
   since $\<spwn> (e'_{2}\Subst{x := e_{1}}) = (\<spwn> e'_{2})\Subst{x := e_{1}}\linebreak = e_{2}\Subst{x := e_{1}} $.

  \item[\inflabel{T-Svc}, \inflabel{T-Req}, \inflabel{T-Img}, \inflabel{T-Snap}, \inflabel{T-Repl} :] Straight\-for\-ward application of (IH) and substitution.
   
  \item[\inflabel{T-TAbs}:] Therefore $e_{2} = \Tabs{\alpha\subtpe T'_{2}}{e'_{2}}$, $T_{2} = \Tuni{\alpha\subtpe T'_{2}}{T''_{2}}$ and
     $\Tjudge{\Gamma, x : T_{1}, \alpha \subtpe T'_{2}\mid\Sigma}{e'_{2}}{T''_{2}}$.
     By the variable convention, it holds that $\alpha$ is not free in $T_{1}$.
     Therefore, by permutation we obtain a derivation of  $\Tjudge{\Gamma, \alpha \subtpe T'_{2}, x : T_{1}\mid\Sigma}{e'_{2}}{T''_{2}}$.
     Together with (IH) we obtain a derivation for $\Tjudge{\Gamma, \alpha \subtpe T'_{2}\mid\Sigma}{e'_{2}\Subst{x := T_{1}}}{T''_{2}}$.
     Extending this derivation with rule \inflabel{T-TAbs}, we obtain a derivation of 
     $\Tjudge{\Gamma\mid\Sigma}{\Tabs{\alpha \subtpe T'_{2}}{e'_{2}\Subst{x := T_{1}}}}{T_{2}}$, which
     is also a derivation of $\Tjudge{\Gamma\mid\Sigma}{e_{2}\Subst{x := e_{1}}}{T_{2}}$ as desired.

  \item[\inflabel{T-TApp}:] Therefore $e_{2} = \Tapp{e'_{2}}{T'_{2}}$, $T_{2} = T''_{2}\Subst{\alpha := T'_{2}}$, $\ftv{T'_{2}}\subseteq \ftv{\Gamma}$, 
$\Subjudge{\Gamma, x\tpe T_{1}}{T'_{2}}{T'''_{2}}$ and   $\Tjudge{\Gamma, x\tpe T_{1}\mid\Sigma}{e'_{2}}{\Tuni{\alpha\subtpe T'''_{2}}{T''_{2}}}$.
Applying (IH) yields a derivation of  $\Tjudge{\Gamma\mid\Sigma}{e'_{2}\Subst{x := e_{1}}}{\Tuni{\alpha\subtpe T'''_{2}}{T''_{2}}}$. 
Note that $\Subjudge{\Gamma, x\tpe T_{1}}{T'_{2}}{T'''_{2}}$ implies $\Subjudge{\Gamma}{T'_{2}}{T'''_{2}}$, since
assumptions on variables do not play a role in subtyping rules.
With rule \inflabel{T-TApp} we obtain a derivation of $\Tjudge{\Gamma\mid\Sigma}{\Tapp{e'_{2}\Subst{x := e_{1}}}{T'_{2}}}{T_{2}}$,
which is also a derivation of $\Tjudge{\Gamma\mid\Sigma}{e_{2}\Subst{x := e_{1}}}{T_{2}}$ as desired.

  \item[\inflabel{T-Sub}:] By premise of the rule, $\Tjudge{\Gamma, x\tpe T_{1}\mid\Sigma}{e_{2}}{T'_{2}}$ and $\Subjudge{\Gamma,x\tpe T_{1}}{T'_{2}}{T_{2}}$.
   The latter implies $\Subjudge{\Gamma}{T'_{2}}{T_{2}}$,  since
assumptions on variables are not required in subtyping rules.
 Apply (IH) to obtain a derivation of $\Tjudge{\Gamma\mid\Sigma}{e_{2}\Subst{x := e_{1}}}{T'_{2}}$. Together
with the previously established facts and \inflabel{T-Sub} we
   obtain a derivation of $\Tjudge{\Gamma\mid\Sigma}{e_{2}\Subst{x := e_{1}}}{T_{2}}$ as desired.
  \end{description}
\end{proof}

\begin{lemma}[Type Substitution Preserves Subtyping]\label{lem:typesubsubst}\ \\
  If $\Subjudge{\Gamma, \alpha\subtpe T', \Gamma'}{S}{T}$ and $\Subjudge{\Gamma}{S'}{T'}$
  then \\ $\Subjudge{\Gamma, \Gamma'\sigma}{S\sigma}{T\sigma}$ where 
  $\sigma = \Subst{\alpha := S'}$.
\end{lemma}
\begin{proof}
  By induction on the typing derivation $\mathcal{D}$ of $\Subjudge{\Gamma, \alpha\subtpe T', \Gamma'}{S}{T}$. In each
  case we assume $\Subjudge{\Gamma}{S'}{T'}$ and $\sigma = \Subst{\alpha := S'}$.
  \newline
  
  \noindent \textbf{Basis:}

  \begin{description}
  \item[\inflabel{S-Top}:] Therefore $T = \Top$.
    By rule \inflabel{S-Top} it holds that $\Subjudge{\Gamma, \Gamma'\sigma}{S\sigma}{\Top}$, i.e., $\Subjudge{\Gamma, \Gamma'\sigma}{S\sigma}{T\sigma}$ as desired.
  \item[\inflabel{S-Refl}:]
    Therefore  $S = T$. $\Subjudge{\Gamma, \Gamma'\sigma}{S\sigma}{T\sigma}$ holds by rule \inflabel{S-Refl}.
  \item[\inflabel{S-TVar}:] 
    Therefore $S = \alpha'$, $\alpha' \subtpe T \in (\Gamma, \alpha \subtpe T', \Gamma')$.
    Case distinction:
    \begin{description}
    \item[$\alpha' \neq \alpha$:] Immediate by rule \inflabel{S-TVar}.
    \item[$\alpha' = \alpha$:] Therefore $S = \alpha$, $S\sigma = T'$, and  $T' = T = T\sigma$. 
      Apply rule \inflabel{S-Refl}.
    \end{description}
  \item[\inflabel{S-Srv$_{\bot}$}:] Immediate by rule \inflabel{S-Srv$_{\bot}$}.
  \end{description}
  
  \noindent \textbf{Inductive step:}

  \noindent \textit{Induction hypothesis (IH)}: The property holds for all proper subderivations of the derivation $\Subjudge{\Gamma, \alpha\subtpe T', \Gamma'}{S}{T}$.

  \begin{description}
  \item[\inflabel{S-Srv}:]
    Therefore $S = \srvt{\tuple{x\tpe S_{2}}}$, $T = \srvt{\tuple{y\tpe T_{2}}}$ and 
    for each $j$ there is $i$ such that $y_{j} = x_{i}$  and  $\Subjudge{\Gamma, \alpha \subtpe T', \Gamma'}{S_{2,i}}{T_{2,j}}$.
    Applying (IH) yields $\Subjudge{\Gamma, \Gamma'\sigma}{S_{2,i}\sigma}{T_{2,j}\sigma}$.
    Together with rule \inflabel{S-Srv} we obtain a derivation of $\Subjudge{\Gamma,\Gamma'\sigma}{\srvt{\tuple{x\tpe S_{2}\sigma}}}{\srvt{\tuple{y\tpe T_{2}\sigma}}}$,
    i.e., $\Subjudge{\Gamma,\Gamma'\sigma}{S\sigma}{T\sigma}$ as desired.

  \item[\inflabel{S-Inst}, \inflabel{S-Img}, \inflabel{S-Svc}, \inflabel{S-Trans}:]
    Straightforward application of (IH).

  \item[\inflabel{S-Univ}:]
    Therefore $S = \Tuni{\alpha_{1}\subtpe U}{S_{2}}$, $T = \Tuni{\alpha_{2}\subtpe U}{T_{2}}$ and
    $\Subjudge{\Gamma, \alpha\subtpe T', \Gamma', \alpha_{1}\subtpe U}{S_{2}}{T_{2}\Subst{\alpha_{2} := \alpha_{1}}}$.
    Together with (IH) we obtain a derivation of 
    $\Subjudge{\Gamma, \Gamma'\sigma, \alpha_{1}\subtpe U\sigma}{S_{2}\sigma}{T_{2}\Subst{\alpha_{2} := \alpha_{1}}\sigma}$, i.e.,
    $\Subjudge{\Gamma, \Gamma'\sigma, \alpha_{1}\subtpe U\sigma}{S_{2}\sigma}{T_{2}\sigma\Subst{\alpha_{2} := \alpha_{1}}}$ (since by our variable convention, we
    may assume $\alpha_{2} \neq \alpha$ and $\alpha_{1}\neq \alpha$). 
    Together with rule \inflabel{S-Univ} we obtain a derivation of
    $\Subjudge{\Gamma, \Gamma'\sigma}{\Tuni{\alpha_{1}\subtpe U\sigma}{S_{2}\sigma}}{ \Tuni{\alpha_{2}\subtpe U\sigma}{T_{2}\sigma}     }$, i.e.,
    $\Subjudge{\Gamma,\Gamma'\sigma}{S\sigma}{T\sigma}$ as desired.
  \end{description}
\end{proof}

 \begin{lemma}[Type Substitution Lemma]\label{lem:typesubst}
   If $\Tjudge{\Gamma, \alpha\subtpe S,\Gamma'\mid\Sigma}{e}{T}$ and $\Subjudge{\Gamma}{S'}{S}$ then $\Tjudge{\Gamma, \Gamma'\sigma\mid\Sigma\sigma}{e\sigma}{T\sigma}$
   where $\sigma = \Subst{\alpha := S'}$.
 \end{lemma}
\begin{proof}
  By induction on the typing derivation $\mathcal{D}$ of $\Tjudge{\Gamma, \alpha\subtpe S, \Gamma'\mid\Sigma}{e}{T}$. In each
  case we assume $\Subjudge{\Gamma}{S'}{S}$ and $\sigma = \Subst{\alpha := S'}$.
  \newline
  
  \noindent \textbf{Basis:}

  \begin{description}
  \item[\inflabel{T-Var}:] Therefore $e = x$ for $x\in\Names\cup\Set{\<this>}$ and $(\Gamma, \alpha \subtpe S, \Gamma')(x) = T$.
    By the variable convention, it holds that  $\alpha$ is not bound in $\Gamma$, therefore
    $(\Gamma, \Gamma'\sigma)(x) = T\sigma$, which holds by a structural induction on $\Gamma'$.
    Together with \inflabel{T-Var} we obtain $\Tjudge{\Gamma,\Gamma'\sigma\mid\Sigma\sigma}{e\sigma}{T\sigma}$ as desired.

  \item[\inflabel{T-Inst}:] Immediate, since $\Gamma$ is not considered in the premises.
  \item[\inflabel{T-$\dead$}:] Immediate by rule \inflabel{T-$\dead$}.
   
  \end{description}
  
  \noindent \textbf{Inductive step:}

  \noindent \textit{Induction hypothesis (IH)}: The property holds for all proper subderivations of the derivation $\mathcal{D}$ of $\Tjudge{\Gamma, \alpha\subtpe S,\Gamma'\mid\Sigma}{e}{T}$.

  \begin{description}
  \item[\inflabel{T-Par}:] From the conclusion of the rule it holds that $e = \<par> \tuple{e_{2}}$, $T = \Unit$
    and from its premises $\Tjudge{\Gamma,\alpha\subtpe S,\Gamma'\mid\Sigma}{e_{2,i}}{T}$ for each $e_{2,i}$ in the sequence $\tuple{e_{2}}$.
    Applying (IH) yields $\Tjudge{\Gamma,\Gamma'\sigma\mid\Sigma\sigma}{e_{2,i}\sigma}{T\sigma}$ for each $i$. Together with $T\sigma = \Unit\sigma = \Unit = T$,
    by \inflabel{T-Par} we obtain a derivation of $\Tjudge{\Gamma,\Gamma'\sigma\mid\Sigma\sigma}{\<par>\tuple{e_{2}\sigma}}{T}$, which is also
    a derivation of $\Tjudge{\Gamma,\Gamma'\mid\Sigma\sigma}{e\sigma}{T\sigma}$ as desired.

  \item[\inflabel{T-Srv}:] Therefore $e = \srvt{\tuple{r}}$, $r_{i} = \tuple{p}_{i}\trr e_{i}$, $T = \Tsrv{\tuple{x_{i,j}\tpe S_{i,j}}}$,
     $\ftv{T} \subseteq \ftv{\Gamma, \alpha\subtpe S, \Gamma'}$,  
     $p_{i,j} = x_{i,j}<\tuple{y_{i,j}:T_{i,j}} >$, $S_{i,j} = <\tuple{T_{i,j}}>$ and $\Tjudge{\Gamma, \alpha\subtpe S, \Gamma' , \tuple{y_{i,j}\tpe T_{i,j}}, \<this>\tpe T\mid\Sigma}{e_{i}}{\Unit}$
     for each $r_{i}$ in $\tuple{r}$.
     Applying (IH) yields $\Tjudge{\Gamma,\Gamma'\sigma, \tuple{y_{i,j}\tpe T_{i,j}\sigma}, \<this>\tpe T\sigma \mid\Sigma\sigma    }{e_{i}\sigma}{\Unit}$ for each $i$ in $\tuple{r}_{i}$.
     Note that $\ftv{T} \subseteq \ftv{\Gamma, \alpha\subtpe S, \Gamma'}$, $\Subjudge{\Gamma}{S'}{S}$ and $\sigma = \Subst{\alpha := S'}$ imply
     $\ftv{T\sigma}\subseteq \ftv{\Gamma, \Gamma'\sigma}$.
     By applying $\sigma$ to the types in the assumptions of the original derivation $\mathcal{D}$, we obtain with the previously established facts
     a derivation of $\Tjudge{\Gamma,\Gamma'\sigma\mid\Sigma\sigma}{e\sigma}{T\sigma}$ as desired by rule \inflabel{T-Srv}.

  \item[\inflabel{T-Img}, \inflabel{T-Snap}, \inflabel{T-Repl}, \inflabel{T-Spwn},\inflabel{T-Svc},\inflabel{T-Req}, \inflabel{T-TAbs}:] Straightforward application of (IH) and substitution.

  \item[\inflabel{T-TApp}:] Therefore $e = \Tapp{e_{2}}{T_{2}}$, $T = T'\Subst{\alpha' := T_{2}}$, $\ftv{T_{2}}\subseteq \ftv{\Gamma \alpha \subtpe S, \Gamma'}$, 
    $\Subjudge{\Gamma, \alpha\subtpe S, \Gamma'}{T_{2}}{T_{3}}$ and   $\Tjudge{\Gamma, \alpha\subtpe S, \Gamma'\mid\Sigma}{e_{2}}{\Tuni{\alpha'\subtpe T_{3}}{T'}}$.
    Applying (IH) yields  $\Tjudge{\Gamma, \Gamma'\sigma\mid\Sigma\sigma}{e_{2}\sigma}{(\Tuni{\alpha'\subtpe T_{3}}{T'})\sigma}$,
    hence $\Tjudge{\Gamma, \Gamma'\sigma\mid\Sigma\sigma}{e_{2}\sigma}{\Tuni{\alpha'\subtpe T_{3}\sigma}{T'\sigma}}$. By lemma~\ref{lem:typesubsubst}
    and $\Subjudge{\Gamma, \alpha\subtpe S, \Gamma'}{T_{2}}{T_{3}}$ it holds that $\Subjudge{\Gamma,  \Gamma'\sigma}{T_{2}\sigma}{T_{3}\sigma}$.
    From $\ftv{T_{2}}\subseteq \ftv{\Gamma, \alpha \subtpe S, \Gamma'}$ it holds that $\ftv{T_{2}\sigma}\subseteq \ftv{\Gamma, \Gamma'\sigma}$.
    Applying rule \inflabel{T-TApp} yields a derivation of $\Tjudge{\Gamma,\Gamma'\sigma\mid\Sigma\sigma}{\Tapp{e_{2}\sigma}{T_{2}\sigma}}{(T'\sigma)\Subst{\alpha' := T_{2}\sigma}}$,
    i.e., $\Tjudge{\Gamma,\Gamma'\sigma\mid\Sigma\sigma}{(\Tapp{e_{2}}{T_{2}})\sigma}{(T'\Subst{\alpha' := T_{2}})\sigma}$, i.e., $\Tjudge{\Gamma,\Gamma'\sigma\mid\Sigma\sigma}{e\sigma}{T\sigma}$ as desired.

  \item[\inflabel{T-Sub}:] By premise of the rule, $\Tjudge{\Gamma, \alpha\subtpe S, \Gamma'\mid\Sigma}{e}{T'}$ and $\Subjudge{\Gamma,\alpha\subtpe S, \Gamma'}{T'}{T}$.
    Applying (IH) to the former yields $\Tjudge{\Gamma,\Gamma'\sigma\mid\Sigma\sigma}{e\sigma}{T'\sigma}$. By lemma~\ref{lem:typesubsubst}
    and $\Subjudge{\Gamma,\alpha\subtpe S, \Gamma'}{T'}{T}$ it holds that $\Subjudge{\Gamma,\Gamma'\sigma}{T'\sigma}{T\sigma}$.
    Thus, by rule \inflabel{T-Sub} we obtain a derivation of $\Tjudge{\Gamma,\Gamma'\sigma\mid\Sigma\sigma}{e\sigma}{T\sigma}$ as desired.
  \end{description}
\end{proof}

\begin{lemma}[Location Typing Extension Preserves Types]\label{lem:location-typing}\ \\
  If $\Tjudge{\Gamma\mid\Sigma}{e}{T}$ and $\Sigma\subseteq\Sigma'$, then $\Tjudge{\Gamma\mid\Sigma'}{e}{T}$.
\end{lemma}
\begin{proof}
 Straighforward induction on the derivation for $\Tjudge{\Gamma\mid\Sigma}{e}{T}$.
\end{proof}

\begin{theorem}[Preservation]
   If $\Tjudge{\Gamma\mid\Sigma}{e}{T}$ and  $\Gamma\mid\Sigma\vdash\mu$ and $e\mid\nobreak\mu --> e'\mid\nobreak\mu'$, then $\Tjudge{\Gamma\mid\Sigma'}{e'}{T}$ for some $\Sigma'$, where $\Sigma\subseteq\Sigma'$ and $\Gamma\mid\Sigma'\vdash \mu'$.
\end{theorem}
\begin{proof}
  By induction on the typing derivation $\mathcal{D}$ of $\Tjudge{\Gamma\mid\Sigma}{e}{T}$. We always assume
  $e\mid\mu --> e'\mid\mu'$ for some $e'$, $\mu'$. Otherwise, $e$ is stuck ($e\not\longrightarrow$) and the property trivially holds.
  We also assume $\Gamma\mid\Sigma\vdash\mu$ in each case.
  \newline
  
  \noindent \textbf{Basis:}

  \begin{description}
  \item[\inflabel{T-Var}, \inflabel{T-Inst}, \inflabel{T-$\dead$}:] Immediate, since $e$ is stuck.
  \end{description}
  
  \noindent \textbf{Inductive step:}

  \noindent \textit{Induction hypothesis (IH)}: The property holds for all proper subderivations of the derivation $\mathcal{D}$ of $\Tjudge{\Gamma\mid\Sigma}{e}{T}$.

  \begin{description}
  \item[\inflabel{T-Srv}, \inflabel{T-Img}, \inflabel{T-TAbs}:] The property trivially holds, since in each case,
   $e$ is stuck.

  \item[\inflabel{T-Par}:] From the conclusion of the rule it holds that $e = \<par> \tuple{e_{1}}$, $T = \Unit$
    and from its premises $\Tjudge{\Gamma\mid\Sigma}{e_{1,i}}{\Unit}$ for each $e_{1,i}$ in the sequence $\tuple{e_{1}}$.
    By the structure of $e$, there are three possible rules which can be at the root of the derivation for $e\mid\mu --> e'\mid\mu'$:
    \begin{description}
    \item[\inflabel{Par}:] Therefore $e = \<par> \tuple{e_{11}}\; (\<par> \tuple{e_{12}})\; \tuple{e_{13}}$ and
      $e' = \<par> \tuple{e_{11}}\;\tuple{e_{12}}\;\tuple{e_{13}}$ and $\mu' = \mu$. From the premises of \inflabel{T-Par}
      it holds that $\Tjudge{\Gamma\mid\Sigma}{e_{12,k}}{\Unit}$ for each $e_{12,k}$ in the sequence $\tuple{e_{12}}$.
      Choose $\Sigma' = \Sigma$. Together with the previously established facts we obtain a derivation of $\Tjudge{\Gamma\mid\Sigma'}{\<par> \tuple{e_{11}}\;\tuple{e_{12}}\;\tuple{e_{13}}}{T}$
      by \inflabel{T-Par}, $\Sigma\subseteq\Sigma'$ and $\Gamma\mid\Sigma'\vdash \mu'$ as desired.

    \item[\inflabel{React}:] Therefore $e = \<par> e''$, $e' = \prlll{e''\;\sigma_{b}(e_{b})}$, $\mu' = \mapadd{\mu}{i \mapsto (s, \tuple{m}')}$.
     By the premises of \inflabel{React}, $\mu(i) = (\srvt{\tuple{r}_{1}\;(\tuple{p}\trr e_{b})\;\tuple{r}_{2}}, \tuple{m})$,
     $\matchp{\tuple{p}, \tuple{m}}{\tuple{m}', \sigma}$ and $\sigma_{b} = \sigma \cup \Subst{\<this> := i}$. Choose $\Sigma' = \Sigma$.
     Since $\Gamma\mid\Sigma\vdash\mu$, $\mu(i)$ is well typed. From its shape it is typed by rule \inflabel{T-Img} as some $\Timg{T'}$.
     Thus, by the premises of \inflabel{T-Img}, $s$ is typed as $\Tsrv{T'}$ and each element in the buffer $\tuple{m}$ 
     is a valid request value for the server template $s$.
     By the match soundness and completeness lemma from the paper, each request value in the buffer $\tuple{m}'$ occurs in $\tuple{m}$.
     Hence, we obtain a derivation for \Tjudge{\Gamma,\Sigma}{(s, \tuple{m}')}{\Timg{T'}}. Thus, $\Gamma\mid\Sigma\vdash \mu'$ 
     and also $\Gamma\mid\Sigma'\vdash\mu'$.

     From the shape of server template $s$, it must be typed by rule \inflabel{T-Srv} as the last step in a derivation $D_{s}$. Therefore, from the premises of this
     rule, we obtain a derivation for $\Tjudge{\Gamma, \tuple{y_{l,k}\tpe T_{l,k}}, \<this>\tpe \Tsrv{T'}\mid\Sigma}{e_{b}}{\Unit}$,
     i.e., $\Tjudge{\Gamma, \tuple{y_{l,k}\tpe T_{l,k}}, \<this>\tpe \Tsrv{T'}\mid\Sigma'}{e_{b}}{\Unit}$, since $\Sigma' = \Sigma$.
     The $\tuple{y_{l,k}\tpe T_{l,k}}$ are the arguments in the join pattern $\tuple{p}$. Applying the substitution lemma~\ref{lem:substition}
     multiple times to the latter derivation yields a derivation of $\Tjudge{\Gamma\mid\Sigma'}{e_{b}}{\Unit}$.
     $\Tjudge{\Gamma\mid\Sigma'}{\sigma_{b}(e_{b})}{\Unit}$. The first application of the lemma to eliminate $\<this>\tpe \Tsrv{T'}$
     is justified by the derivation $\mathcal{D}_{s}$. The other applications to eliminate $\tuple{y_{l,k}\tpe T_{l,k}}$ are justified by
     the match soundness and completeness lemma, which guarantees that the selection of argument values in the substitution $\sigma$ are from
     matching service request values in $\tuple{m}$, which is well-typed under $\Gamma\mid\Sigma'$. Hence $\tuple{\Tjudge{\Gamma\mid\Sigma'}{\sigma(y_{l,k})}{T_{l,k}}}$ holds.

     Finally, since $\Sigma' = \Sigma$, $\Tjudge{\Gamma\mid\Sigma}{e''}{\Unit}$, we have $\Tjudge{\Gamma\mid\Sigma'}{e''}{\Unit}$.
     Together with $\Tjudge{\Gamma\mid\Sigma'}{e_{b}}{\Unit}$ by rule \inflabel{T-Par}, we obtain $\Tjudge{\Gamma\mid\Sigma'}{\prlll{e''\;e_{b}}}{\Unit}$,
     i.e., $\Tjudge{\Gamma\mid\Sigma'}{e'}{T}$. This together with the previously established $\Gamma\mid\Sigma'\vdash \mu'$ is the property we wanted
     to show.          

    \item[\inflabel{Cong}:] Therefore $e = \context{E}{e_{1,j}}$ for an expression $e_{1,j}$ in the sequence $\tuple{e_{1}}$ and
     $e' = \context{E}{e'_{1,j}}$ for some $e'_{1,j}$, where $e_{1,j}\mid\mu --> e'_{1,j}\mid\mu'$.
     Since $\Tjudge{\Gamma\mid\Sigma}{e_{1,j}}{\Unit}$, it follows from (IH) that there is a derivation of $\Tjudge{\Gamma\mid\Sigma'}{e'_{1,j}}{\Unit}$
     with $\Sigma\subseteq\Sigma'$ and $\Gamma\mid\Sigma\vdash\mu'$
     From $\Tjudge{\Gamma\mid\Sigma}{e_{1,i}}{\Unit}$ for each $i\neq j$, $\Sigma\subseteq\Sigma'$ and lemma~\ref{lem:location-typing},
     we obtain derivations $\Tjudge{\Gamma\mid\Sigma'}{e_{1,i}}{\Unit}$. Together with $\Tjudge{\Gamma\mid\Sigma'}{e'_{1,j}}{\Unit}$     
     we obtain a derivation for \Tjudge{\Gamma\mid\Sigma'}{e'}{T} by rule \inflabel{T-Par} as desired.
    \end{description}

  \item[\inflabel{T-Snap}:] Therefore $e = \<snap> e_{1}$, $T = \Timg{T'}$ and \Tjudge{\Gamma\mid\Sigma}{e_{1}}{\Tsrvinst{T'}}. By the structure of $e$, there are two possible rules which can be at the root of the derivation for $e\mid\mu --> e'\mid\mu'$:
    \begin{description}
    \item[\inflabel{Snap}:] 
      Therefore, $e_{1} = i \in \Nat$, $e' = (\srvt{\tuple{r}}, \tuple{m})$ or $e' = \dead$, and $\mu' = \mu$. From \Tjudge{\Gamma\mid\Sigma}{e_{1}}{\Tsrvinst{T'}} and the shape of $e_{1}$, rule \inflabel{T-Inst} is the root
      of the corresponding derivation. Thus, $i \in \Sigma$ and $\Sigma(i) = \Tsrvinst{T'}$ by the premises of this rule. Choose $\Sigma' = \Sigma$. Since $\Gamma\mid\Sigma\vdash\mu$, $\mu = \mu'$ and $\Sigma' = \Sigma$,
      it also holds that $\Gamma\mid\Sigma'\vdash\mu'$. Hence $\Tjudge{\Gamma\mid\Sigma'}{e'}{T}$ as desired.

    \item[\inflabel{Cong}:] 
      Therefore, $e' = \<snap> e_{2}$ for some $e_{2}$, and $e_{1}\mid\mu --> e_{2}\mid\mu'$ holds. Applying this together with \Tjudge{\Gamma{\mid\Sigma}}{e_{1}}{\Tsrvinst{T'}}
      to (IH) yields \Tjudge{\Gamma\mid\Sigma''}{e_{2}}{\Tsrvinst{T'}} for $\Sigma''$, where $\Sigma\subseteq\Sigma''$ and $\Gamma\mid\Sigma''\vdash \mu'$.
      Together with rule \inflabel{T-Snap} we obtain \Tjudge{\Gamma\mid\Sigma''}{\<snap> e_{2}}{\Timg{T'}}, i.e., \Tjudge{\Gamma\mid\Sigma''}{e'}{T}.
      Choose $\Sigma' = \Sigma''$.      
    \end{description}

  \item[\inflabel{T-Repl}:] Therefore $e = \repl{e_{1}\;e_{2}}$, $T = \Unit$,\linebreak \Tjudge{\Gamma{\mid\Sigma}}{e_{1}}{\Tsrvinst{T'}} and \Tjudge{\Gamma{\mid\Sigma}}{e_{2}}{\Timg{T'}}. 
    By the structure of $e$, there are two possible rules which can be at the root of the derivation for $e\mid\mu --> e'\mid\mu'$:
    \begin{description}
    \item[\inflabel{Repl}:] 
      Therefore, $e_{1} = i \in \Nat$, $i\in\dom{\mu}$, $e_{2} = (\srvt{\tuple{r}}, \tuple{m})$ or $e_{2} = \dead$, $e' = \prlll{\varepsilon}$ and $\mu' = \mapadd{mu}{i \mapsto s}$. 
      Hence $i$ is well typed under $\Gamma\mid\Sigma$ as $\Tsrvinst{T'}$. Together with \Tjudge{\Gamma{\mid\Sigma}}{e_{2}}{\Timg{T'}} and $\Gamma\mid\Sigma\vdash \mu$
      it holds that $\Gamma\mid\Sigma\vdash\mu'$. Choose $\Sigma' = \Sigma$. Apply rule \inflabel{T-Par} to obtain \Tjudge{\Gamma\mid\Sigma'}{e'}{T} as desired.

    \item[\inflabel{Cong}:] 
      Therefore, $e' = \<repl> e'_{1}\; e'_{2}$ for some $e'_{1}$, $e'_{2}$ and either $e_{1}\mid\mu --> e'_{1}\mid\mu'$, $e_{2} = e'_{2}$ or $e_{2}\mid\mu --> e'_{2}\mid\mu'$, $e_{1} = e'_{1}$ holds. 
      We only show the first case, the other is similar.
      Apply (IH) to \Tjudge{\Gamma{\mid\Sigma}}{e_{1}}{\Tsrvinst{T'}} and $e_{1}\mid\mu --> e'_{1}\mid\mu'$ to obtain $\Sigma''$ with $\Sigma\subseteq\Sigma''$ and $\Gamma\mid\Sigma''\vdash\mu'$
      and \Tjudge{\Gamma\mid\Sigma''}{e'_{1}}{\Tsrvinst{T'}}. Choose $\Sigma' = \Sigma''$. Apply lemma~\ref{lem:location-typing} to \Tjudge{\Gamma{\mid\Sigma}}{e_{2}}{\Timg{T'}} in order 
      to obtain \Tjudge{\Gamma{\mid\Sigma'}}{e_{2}}{\Timg{T'}}. Finally, apply rule \inflabel{T-Repl} to obtain \Tjudge{\Gamma\mid\Sigma'}{e'}{T} as desired.      
    \end{description}

  \item[\inflabel{T-Spwn}:] Therefore $e = \<spwn> e''$, $T = \Tsrvinst{T'}$ and $\Tjudge{\Gamma\mid\Sigma}{e''}{\Timg{T'}}$.
    By the structure of $e$, there are two possible rules which can be at the root of the derivation of $e --> e'$:
    \begin{description}
    \item[\inflabel{Spwn}:] Therefore $e'' = (\srvt{\tuple{r}}, \tuple{m})$ or $e'' = \dead$, $e' = i \in \Nat$, $i\notin\dom{\mu}$ and $\mu' = \mapadd{\mu}{i \mapsto e''}$.
      With $\Gamma\mid\Sigma\vdash \mu$ and definition~\ref{def:well-typed-rt} it follows that $i\notin\dom{\Sigma}$.
      Choose $\Sigma' = \mapadd{\Sigma}{i \mapsto \Timg{T'}}$. By rule \inflabel{T-Inst} and definition of $\Sigma'$, it holds that $\Tjudge{\Gamma\mid\Sigma'}{i}{\Tsrvinst{T'}}$,
      i.e., $\Tjudge{\Gamma\mid\Sigma'}{e'}{T}$. By construction, $\Sigma \subseteq \Sigma'$. 

      What is left to show is $\Gamma\mid\Sigma'\vdash\mu'$:

      First note $\dom{\mu} = \dom{\Sigma}$ and for all $j \in \dom{\Sigma}$,
      $\Sigma(j) = \Sigma'(j)$, $\mu(j) = \mu'(j)$  and $\Tjudge{\Gamma\mid\Sigma}{\mu(j)}{\Sigma(j)}$.
      Hence $\Tjudge{\Gamma\mid\Sigma}{\mu'(j)}{\Sigma'(j)}$ and by lemma~\ref{lem:location-typing}, 
      $\Tjudge{\Gamma\mid\Sigma'}{\mu'(j)}{\Sigma'(j)}$ for each $j \in \dom{\Sigma}$.
      
      From $\Tjudge{\Gamma\mid\Sigma}{e''}{\Timg{T'}}$ , $\Sigma \subseteq \Sigma'$ and lemma~\ref{lem:location-typing} we obtain
      $\Tjudge{\Gamma\mid\Sigma'}{e''}{\Timg{T'}}$. Together with the definitions of $\Sigma'$, $\mu'$, this is a derivation for
      $\Tjudge{\Gamma\mid\Sigma'}{\mu'(i)}{\Sigma'(i)}$. 
      
      In summary, we have established that\linebreak $\Tjudge{\Gamma\mid\Sigma'}{\mu'(j)}{\Sigma'(j)}$ for all $j \in \dom{\mu}\cup\Set{i} = \dom{\mu'} = \dom{\Sigma'}$.
      By definition~\ref{def:well-typed-rt}, this means $\Gamma\mid\Sigma'\vdash\mu'$, what was left to show.
            
    \item[\inflabel{Cong}:] Therefore, by the structure of $e$, it holds that $e = \context{E}{e''}$ for the context $\context{E}{\cdot} = \<spwn>\hole$. By the premise of \inflabel{Cong}
     we obtain $e''\mid\mu --> e'''\mid\mu'$, hence $e' = \context{E}{e'''} = \<spwn> e'''$. Applying the (IH) to $\Tjudge{\Gamma\mid\Sigma}{e''}{\Timg{T'}}$ and  $e''\mid\mu --> e'''\mid\mu'$     yields a derivation of
     $\Tjudge{\Gamma\mid\Sigma'}{e'''}{\Timg{T'}}$ for some $\Sigma'$ with $\Sigma\subseteq\Sigma'$ and $\Gamma\mid\Sigma'\vdash\mu'$. 
     From the previous typing derivation and rule \inflabel{T-Spwn} we obtain a derivation for $\Tjudge{\Gamma\mid\Sigma'}{\<spwn> e'''}{\Tsrvinst{T'}}$, i.e.,
     $\Tjudge{\Gamma\mid\Sigma'}{e'}{T}$ as desired.
    \end{description}

  \item[\inflabel{T-Svc}:] Therefore $e = \svc{e''}{x_{i}}$, $T = T_{1,i}$, \linebreak$\Tjudge{\Gamma\mid\Sigma}{e''}{\Tsrvinst{\Tsrv{\tuple{x\tpe T_{1}}}}}$, where
    $x_{i}\tpe T_{1,i}$ occurs in the sequence $\tuple{x\tpe T_{1}}$. Since $e\mid\mu --> e'\mid\mu'$ by assumption, $e''$ cannot be a value, otherwise $e$ too would be a value and hence stuck.
    Together with the structure of $e$,  reduction rule \inflabel{Cong} is the only possible root of the derivation of  $e\mid\mu --> e'\mid\mu'$, where $e = \context{E}{e''}$.
    Hence $e' = \context{E}{e'''} = \svc{e'''}{x_{i}}$ and  $e''\mid\mu --> e'''\mid\mu'$ by the premise of \inflabel{Cong}. Applying the (IH) to $e''\mid\mu --> e'''\mid\mu'$ and
    $\Tjudge{\Gamma\mid\Sigma}{e''}{\Tsrvinst{\Tsrv{\tuple{x\tpe T_{1}}}}}$ yields a derivation of $\Tjudge{\Gamma\mid\Sigma'}{e'''}{\Tsrvinst{\Tsrv{\tuple{x\tpe T_{1}}}}}$,
    for some $\Sigma'$ with $\Sigma\subseteq\Sigma'$ and $\Gamma\mid\Sigma'\vdash\mu'$.
    By rule \inflabel{T-Svc} and the previously established facts, we obtain a derivation for $\Tjudge{\Gamma\mid\Sigma'}{\svc{e'''}{x_{i}}}{T_{1,i}}$, which is also
    a derivation of $\Tjudge{\Gamma\mid\Sigma'}{e'}{T}$ as desired.

  \item[\inflabel{T-Req}:] Therefore $e = e''<e_{1}\ldots e_{n}>$, $T = \Unit$, \linebreak$\Tjudge{\Gamma}{e''}{<T_{1}\ldots T_{n}>}$ and $(\Tjudge{\Gamma}{e_{i}}{T_{i}})_{1\in\oto{n}}$.
    Since $e --> e'$ by assumption, there is an expression in the set $\Set{e'', e_{1}, \ldots, e_{n}}$ which is not a value, otherwise $e$ is a value and stuck.
    Together with the structure of $e$, reduction rule \inflabel{Cong} is the only possible root of the derivation of $e --> e'$. Therefore $e = \context{E}{e'''}$,
    where $\context{E}{\cdot} = \hole<e_{1}\ldots e_{n}>$ or $\context{E}{\cdot} = e''<\tuple{e_{11}}\;\hole\; \tuple{e_{22}}>$. For any of the possible shapes of $\context{E}{}$,
    we can straightforwardly apply the (IH) to obtain  a derivation of $\Tjudge{\Gamma}{e'}{T}$ as desired.   

  \item[\inflabel{T-TApp}:] Therefore $e = \Tapp{e''}{T_{1}}$, $T = T'\Subst{\alpha := T_{1}}$, $\ftv{T_{1}}\subseteq \ftv{\Gamma}$, $\Subjudge{\Gamma}{T_{1}}{T_{2}}$ and
   $\Tjudge{\Gamma\mid\Sigma}{e''}{\Tuni{\alpha\subtpe T_{2}}{T'}}$.
    By the structure of $e$, there are two possible rules which can be at the root of the derivation of $e\mid\mu --> e'\mid\mu'$:
    \begin{description}
    \item[\inflabel{TAppAbs}:] Therefore $e'' = \Tabs{\alpha\subtpe T_{2}}{e'''}$ and hence $e' = e'''\Subst{\alpha := T_{1}}$.
      From the structure of $e$, $e''$ and the available rules, 
      there is a proper subderivation in $\mathcal{D}$ of $\Tjudge{\Gamma, \alpha\subtpe T_{2}\mid\Sigma }{e'''}{T'}$.
      Together with $\Subjudge{\Gamma}{T_{1}}{T_{2}}$ and the type substitution lemma~\ref{lem:typesubst}, we obtain a derivation
      for $\Tjudge{\Gamma\mid\Sigma}{e'''\Subst{\alpha := T_{1}}}{T'\Subst{\alpha := T_{1}}}$. Choose $\Sigma' = \Sigma$, then the previous derivation
      also is a derivation of $\Tjudge{\Gamma\mid\Sigma'}{e'}{T}$ as desired.
       
    \item[\inflabel{Cong}:] Straightforward application of the (IH) similar to the previous cases.
    \end{description}

  \item[\inflabel{T-Sub}:] By premise of the rule, $\Tjudge{\Gamma\mid\Sigma}{e}{T'}$ and $\Subjudge{\Gamma}{T'}{T}$. Apply (IH) to the former and then \inflabel{T-Sub} to
   obtain a derivation of $\Tjudge{\Gamma\mid\Sigma'}{e'}{T}$ and an appropriate $\Sigma'$.
  \end{description}
\end{proof}


}
{}

\end{document}